\documentclass[12pt]{article}
\usepackage{fullpage}
\usepackage{amsmath}
\usepackage{amssymb}
\usepackage{amsthm}
\usepackage{bbm}
\usepackage{color}
\usepackage{enumerate}
\usepackage{thm-restate}
\usepackage{thmtools}
\usepackage[margin=1 in]{geometry}
\usepackage[norelsize,ruled,vlined,linesnumbered]{algorithm2e}

\newtheorem{theorem}{Theorem}[section]
\newtheorem{lemma}[theorem]{Lemma}
\newtheorem{definition}[theorem]{Definition}
\newtheorem{corollary}[theorem]{Corollary}

\newtheorem{proposition}[theorem]{Proposition}

\newtheorem{fact}[theorem]{Fact}

\newcommand{\ep}{\epsilon}
\newcommand{\mc}{\mathcal}

\newcommand{\SSS}{\ensuremath{\texttt{SubspaceSparsifier}}\xspace}
\newcommand{\DiffApx}{\ensuremath{\texttt{DiffApx}}\xspace}
\newcommand{\ResApx}{\ensuremath{\texttt{ResApx}}\xspace}
\newcommand{\Dim}{\ensuremath{\text{dim}}\xspace}
\newcommand{\Tr}{\ensuremath{\text{trace}}\xspace}
\newcommand{\Oracle}{\ensuremath{\mc{O}\xspace}}
\newcommand{\SlowOracle}{\ensuremath{\texttt{SlowOracle}\xspace}}
\newcommand{\FastOracle}{\ensuremath{\texttt{FastOracle}\xspace}}
\newcommand{\Split}{\ensuremath{\texttt{Split}\xspace}}
\newcommand{\Unsplit}{\ensuremath{\texttt{Unsplit}\xspace}}
\newcommand{\SketchMatrix}{\ensuremath{\texttt{SketchMatrix}\xspace}}
\newcommand{\RecoverNorm}{\ensuremath{\texttt{RecoverNorm}\xspace}}
\newcommand{\ColumnApx}{\ensuremath{\texttt{ColumnApx}\xspace}}
\newcommand{\Subsample}{\ensuremath{\texttt{Subsample}\xspace}}

\newcommand{\LaplSolve}{\ensuremath{\texttt{LaplSolve}}\xspace}
\newcommand{\kh}[1]{\left(#1\right)}
\def\trace#1{\mathrm{Tr} \left(#1 \right)}
\def\norm#1{\left\| #1 \right\|}
\def\setof#1{\left\{#1  \right\}}
\def\sizeof#1{\left|#1  \right|}

\newcommand{\lambdamin}{\lambda_{\mathrm{min}}}
\newcommand{\lambdamax}{\lambda_{\mathrm{max}}}

\newcommand{\sigmamax}{\sigma_{\mathrm{max}}}
\newcommand{\wmax}{w_{\mathrm{max}}}
\newcommand{\wmin}{w_{\mathrm{min}}}
\newcommand{\xtil}{\tilde{x}}
\newcommand{\rtil}{\tilde{r}}

\newcommand{\eps}{\epsilon}
\def\abs#1{\left|#1  \right|}
\usepackage[colorlinks=false]{hyperref}
\usepackage{appendix}

\title{Spectral Subspace Sparsification}

\author{Huan Li\\%\footnote{Huan Li is with Shanghai Key Laboratory of Intelligent Information Processing, School of Computer Science, Fudan University, Shanghai 200433.}\\
    School of Computer Science, Fudan University\\ \texttt{huanli16@fudan.edu.cn} \and Aaron Schild\footnote{Supported by NSF grant CCF-1553751} \\EECS, UC
Berkeley\\ \texttt{aschild@berkeley.edu} }

\begin{document}
\maketitle

\begin{abstract}
We introduce a new approach to spectral sparsification that approximates the quadratic form of the pseudoinverse of a graph Laplacian restricted to a subspace. We show that sparsifiers with a near-linear number of edges in the dimension of the subspace exist. Our setting generalizes that of Schur complement sparsifiers. Our approach produces sparsifiers by sampling a uniformly random spanning tree of the input graph and using that tree to guide an edge elimination procedure that contracts, deletes, and reweights edges. In the context of Schur complement sparsifiers, our approach has two benefits over prior work. First, it produces a sparsifier in almost-linear time with no runtime dependence on the desired error. We directly exploit this to compute approximate effective resistances for a small set of vertex pairs in faster time than prior work (Durfee-Kyng-Peebles-Rao-Sachdeva '17). Secondly, it yields sparsifiers that are reweighted minors of the input graph. As a result, we give a near-optimal answer to a variant of the Steiner point removal problem.

A key ingredient of our algorithm is a subroutine of independent interest: a near-linear time algorithm that, given a chosen set of vertices, builds a data structure from which we can query a multiplicative approximation to the decrease in the effective resistance between two vertices after identifying all vertices in the chosen set to a single vertex with inverse polynomial additional additive error in near-constant time.

\end{abstract}

\pagebreak

\section{Introduction}

Graph sparsification has had a number of applications throughout algorithms and theoretical computer science. In this work, we loosen the requirements of spectral sparsification and show that this loosening enables us to obtain sparsifiers with fewer edges. Specifically, instead of requiring that the Laplacian pseudoinverse quadratic form is approximated for every vector, we just require that the sparsifier approximates the Laplacian pseudoinverse quadratic form on a subspace:

\begin{definition}[Spectral subspace sparsifiers]
Consider a weighted graph $G$, a vector space $\mc S \subseteq \mathbb{R}^{V(G)}$ that is orthogonal to $\textbf{1}^{V(G)}$, and $\ep\in (0,1)$. For a minor $H$ of $G$ with contraction map $\phi:V(G)\rightarrow V(H)$, let $P\in \mathbb{R}^{V(H)\times V(G)}$ be a matrix with $P_{uv} = \mathbbm{1}[u = \phi(v)]$ for all $u\in V(H), v\in V(G)$. A reweighted minor $H$ of $G$ is called an $(\mc S,\ep)$-\emph{spectral subspace sparsifier} if for all vectors $x\in \mc S$,

$$(1 - \ep)x^T L_G^+ x\le x_H^T L_H^+ x_H\le (1 + \ep)x^T L_G^+ x$$

\noindent where $x_H := P x$.

\end{definition}

\cite{KMST10} also considers a form of specific form of subspace sparsification related to controlling the $k$ smallest eigenvalues of a spectral sparsifier for $\mc S = \mathbb{R}^{V(G)}$. When $\mc S$ is the dimension $|S|-1$ subspace of $\mathbb{R}^{|S|}\times \textbf{0}^{n-|S|}$ that is orthogonal to $\textbf{1}^{V(G)}$, a $(S,\ep)$-spectral subspace sparsifier is a sparsifier for the Schur complement of $G$ restricted to the set of vertices $S$. Schur complement sparsifiers are implicitly constructed in \cite{KS16} and \cite{KLPSS16} by an approximate Gaussian elimination procedure and have been used throughout spectral graph theory. For example, they are used in algorithms for random spanning tree generation \cite{DKPRS17,DPPR17}, approximate maximum flow \cite{MP13}, and effective resistance computation \cite{GHP18,GHP17,DKPRS17}.

Unlike the existing construction of Schur complement sparsifiers \cite{DKPRS17}, our algorithm (a) produces a sparsifier with vertices outside of $S$ and (b) produces a sparsifier that is a minor of the input graph. While (a) is a disadvantage to our approach, it is not a problem in applications, in which the number of edges in the sparsifier is the most relevant feature for performance, as illustrated by our almost-optimal algorithm for $\ep$-approximate effective resistance computation. (b) is an additional benefit to our construction and connects to the well-studied class of Steiner point removal problems \cite{CGH16,EGKRTT14}.

In the Approximate Terminal Distance Preservation problem \cite{CGH16}, one is given a graph $G$ and a set of $k$ vertices $S$. One is asked find a reweighted minor $H$ of $G$ with size $\text{poly}(k)$ for which

$$d_G(u,v)\le d_H(u,v)\le \alpha d_G(u,v)$$

\noindent for all $u,v\in S$ and some small \emph{distortion} $\alpha > 1$. The fact that $H$ is a minor of $G$ is particularly useful in the context of planar graphs. One can equivalently phrase this problem as a problem of finding a minor $H$ in which the $\ell_1$-norm of the $\ell_1$-minimizing flow between any two vertices $s,t\in S$ is within an $\alpha$-factor of the $\ell_1$ norm of the $\ell_1$-minimizing $s-t$ flow in $G$. The analogous problem for $\ell_{\infty}$ norms is the problem of constructing a \emph{flow sparsifier} (with non-$s-t$ demands as well). Despite much work on flow sparsifiers \cite{M09,LM10,CLLM10,MM10,EGKRTT14,C12,AGK14,RST14}, it is still not known whether $\alpha = (1 + \eps)$-flow sparsifiers with size $\text{poly}(k,1/\ep)$ exist, even when the sparsifier is not a minor of the original graph.

\subsection{Our Results}

Our main result is the following:

\begin{theorem}\label{thm:main-existence}
Consider a weighted graph $G$, a $d$-dimensional vector space $\mc S \subseteq \mathbb{R}^{V(G)}$, and $\ep\in (0,1)$. Then an $(\mc S,\ep)$-spectral subspace sparsifier for $G$ with $O\left(\frac{d \log d}{\ep^2}\right)$ edges exists.
\end{theorem}

When $\mc S$ is the maximal subspace of $\mathbb{R}^S \times \textbf{0}^{V(G)\setminus S}$ orthogonal to $\textbf{1}^{V(G)}$ for some set of vertices $S \subseteq V(G)$, $(\mc S,\ep)$-spectral subspace sparsifiers satisfy the same approximation guarantee as Schur complement sparsifiers. The approximation guarantee of a spectral subspace sparsifier $H$ of $G$ is equivalent to saying that for any demand vector $d\in \mc S$, the energy of the $\ell_2$-minimizing flow for $d$ in $H$ is within a $(1 + \ep)$ factor of the energy for the $\ell_2$-minimizing flow for $d$ in $G$. This yields an near-optimal (up to a $\log d$ factor) answer to the $(1+\ep)$-approximate Steiner vertex removal problem for the $\ell_2$ norm. The $\ell_2$ version is substantially different from the $\ell_1$ problem, in which there do not exist $o(k^2)$-size minors that 2-approximate all terminal distances \cite{CGH16}.

Unlike Schur complement sparsifiers, $(\mathbb{R}^S,\ep)$-spectral subspace sparsifiers may contain ``Steiner nodes;'' i.e. vertices outside of $S$. This is generally not relevant in applications, as we illustrate in Section \ref{sec:res-apx}. Allowing Steiner nodes allows us to obtain sparsifiers with fewer edges, which in turn allows us to obtain faster constructions. Specifically, we show the following result:

\begin{theorem}\label{thm:schur-fast}
Consider a weighted graph $G$, a set of vertices $S\subseteq V(G)$, and $\ep\in (0,1)$. Let $\mc T_{rst}(G)$ denote the time it takes to generate a random spanning tree from a distribution with total variation distance at most $1/m^{10}$ from the uniform distribution. Then $(\mathbb{R}^S\times \textbf{0}^{V(G)\setminus S},\ep)$-spectral subspace sparsifier for $G$ with $\min(m,O\left(|S|\frac{\text{polylog}(n)}{\ep^2}\right))$ edges can be constructed in $\mc T_{rst}(G) + O(m\text{polylog}(n))\le m^{1 + o(1)}$ time.
\end{theorem}

This sparsifier has as many edges as the Schur complement sparsifier given in \cite{DKPRS17} but improves on their $\tilde{O}(m + n/\ep^2)$ runtime. An important ingredient in the above construction is an algorithm for multiplicatively approximating changes in effective resistances due to certain modifications of $G$. This algorithm is called with $\delta = \Theta(1)$ in this paper:

\begin{lemma}\label{lem:diff-apx}
Consider a weighted graph $G$, a set of vertices $S\subseteq V(G)$, and $\delta_0,\delta_1\in (0,1)$. There is an $O(m\text{polylog}(n)\log(m/\delta_1)/\delta_0^2)$-time algorithm $\DiffApx(G,S,\delta_0,\delta_1)$ that outputs numbers $\nu_e$ for all $e\in E(G)$ with the guarantee that

$$(1 - \delta_0)\nu_e - \delta_1 \le \frac{b_e^T L_G^+ b_e}{r_e} - \frac{b_e^T L_{G/S}^+ b_e}{r_e} \le (1 + \delta_0)\nu_e + \delta_1$$
\end{lemma}

Finally, we replace the use of Theorem 6.1 in \cite{DKPRS17} with our Theorem \ref{thm:schur-fast} in their improvement to Johnson-Lindenstrauss to obtain a faster algorithm:

\begin{corollary}\label{cor:res-apx}
    Consider a weighted graph $G$, a set of pairs of vertices $P\subseteq V(G)\times V(G)$, and an $\ep\in (0,1)$. There is an $m^{1+o(1)} + \tilde{O}(|P|/\ep^2)$-time algorithm $\ResApx(G,P,\ep)$ that outputs $(1 + \ep)$-multiplicative approximations to the quantities

$$b_{uv}^T L_G^+ b_{uv}$$

for all pairs $(u,v)\in P$.
\end{corollary}

This directly improves upon the algorithm in \cite{DKPRS17}, which takes $O((m + (n + |P|)/\ep^2)\text{polylog}(n))$-time.

\subsection{Technical Overview}

To construct Schur complement sparsifiers, \cite{DKPRS17} eliminates vertices one-by-one and sparsifies the cliques resulting from those eliminations. This approach is fundamentally limited in that each clique sparsification takes $\Omega(1/\ep^2)$ time in general. Furthermore, in the $n+1$ vertex star graph with $n$ vertices $v_1,v_2,\hdots,v_n$ connected to a single vertex $v_{n+1}$, a $(1+\ep)$-approximate Schur complement sparsifier without Steiner vertices for the set $\{v_1,v_2,\hdots,v_n\}$ must contain $\Omega(n/\ep^2)$ edges. As a result, it seems difficult to obtain Schur complement sparsifiers in time less than $\tilde{O}(m + n/\ep^2)$ time using vertex elimination.

Instead, we eliminate edges from a graph by contracting or deleting them. Edge elimination has the attractive feature that, unlike vertex elimination, it always reduces the number of edges. Start by letting $H := G$. To eliminate an edge $e$ from the current graph $H$, sample $X_e\sim \text{Ber}(p_e)$ for some probability $p_e$ depending on $e$, contract $e$ if $X_e = 1$, and delete $e$ if $X_e = 0$.

To analyze the sparsifier produced by this procedure, we set up a matrix-valued martingale and reduce the problem to bounding the maximum and minimum eigenvalues of a random matrix with expectation equal to the identity matrix. The right value for $p_e$ for preserving this matrix in expectation turns out to be the probability that a uniformly random spanning tree of $H$ contains the edge $e$. To bound the variance of the martingale, one can use the Sherman-Morrison rank one update formula to bound the change in $L_H^+$ due to contracting or deleting the edge $e$. When doing this, one sees that the maximum change in eigenvalue is at most a constant times $$\max_{x\in \mc S} \frac{(x^T L_H^+ b_e)^2}{r_e \min(\texttt{lev}_H(e), 1 - \texttt{lev}_H(e)) (x^T L_G^+ x)}$$ where $\texttt{lev}_H(e)$ is the probability that $e$ is in a uniformly random spanning tree of $H$. This quantity is naturally viewed as the quotient of two quantities:

\begin{enumerate}[(a)]
\item The maximum fractional energy contribution of $e$ to any demand vector in $\mc S$'s electrical flow.
\item The minimum of the probabilities that $e$ is in or is not in a uniformly random spanning tree of $H$.
\end{enumerate}

We now make the edge elimination algorithm more specific to bound these two quantities. Quantity (a) is small on average over all edges in $e$ (see Proposition \ref{prop:sum-bound}), so choosing the lowest-energy edge yields a good bound on the maximum change. To get a good enough bound on the stepwise martingale variance, it suffices to sample an edge uniformly at random from the half of edges with lowest energy. Quantity (b) is often not bounded away from 0, but can be made so by modifying the sampling procedure. Instead of contracting or deleting the edge $e$, start by \emph{splitting} it into two parallel edges with double the resistance or two series edges with half the resistance, depending on whether or not $\texttt{lev}_H(e)\le 1/2$. Then, pick one of the halves $e_0$, contract it with probability $p_{e_0}$, or delete it otherwise. This produces a graph in which the edge $e$ is either contracted, deleted, or reweighted. This procedure suffices for proving our main existence result (Theorem \ref{thm:main-existence}). This technique is similar to the technique used to prove Lemma 1.4 of \cite{Sc17}.

While the above algorithm does take polynomial time, it does not take almost-linear time. We can accelerate it by batching edge eliminations together using what we call \emph{steady oracles}. The contraction/deletion/reweight decisions for edges in $H$ during each batch can be made by sampling just one $1/m^{10}$-approximate uniformly random spanning tree, which takes $m^{1+o(1)}$ time. The main remaining difficulty is finding a large set of edges for which quantity (a) does not change much over the course of many edge contractions/deletions. To show the existence of such a set, we exploit electrical flow localization \cite{SRS17}. To find this set, we use matrix sketching and a new primitive for approximating the change in leverage score due to the identification of some set of vertices $S$ (Lemma \ref{lem:diff-apx}), which may be of independent interest. The primitive for approximating the change works by writing the change in an Euclidean norm, reducing the dimension by Johnson-Lindenstrauss Lemma, and then computing the embedding by Fast Laplacian Solvers in near-linear time.

We conclude by briefly discussing why localization is relevant for showing that quantity (a) does not change over the course of many iterations. The square root of the energy contribution of an edge $e$ to $x$'s electrical flow after deleting an edge $f$ is

    \vspace{-5pt}
\begin{align*}
\left|\frac{x^T L_{H\backslash f}^+ b_e}{\sqrt{r_e}}\right| &= \left|\frac{x^T L_H^+ b_e}{\sqrt{r_e}} + \frac{(x^T L_H^+ b_f)(b_f^T L_H^+ b_e)}{(r_f - b_f^T L_H^+ b_f)\sqrt{r_e}}\right|\\
&= \left|\frac{x^T L_H^+ b_e}{\sqrt{r_e}} + \frac{1}{1-\texttt{lev}_H(f)}\frac{x^T L_H^+ b_f}{\sqrt{r_f}}\frac{b_f^T L_H^+ b_e}{\sqrt{r_f}\sqrt{r_e}}\right|\\
&\le \left|\frac{x^T L_H^+ b_e}{\sqrt{r_e}}\right| + \frac{1}{1-\texttt{lev}_H(f)}\left|\frac{x^T L_H^+ b_f}{\sqrt{r_f}}\right|\left|\frac{b_f^T L_H^+ b_e}{\sqrt{r_f}\sqrt{r_e}}\right|\\
\vspace{-10pt}
\end{align*}
\noindent by Sherman-Morrison. In particular, the new energy on $e$ is at most the old energy plus some multiple of the energy on the deleted edge $f$. By \cite{SRS17}, the average value of this multiplier over all edges $e$ and $f$ is $\tilde{O}(\frac{1}{|E(H)|})$, which means that the algorithm can do $\tilde{\Theta}(|E(H)|)$ edge deletions/contractions without seeing the maximum energy on edges $e$ change by more than a factor of 2.

\vspace{.1 in}

\noindent \textbf{Acknowledgements.} We thank Richard Peng, Jason Li, and Gramoz Goranci for helpful discussions.

\newpage
\tableofcontents

\newpage

\section{Preliminaries}

\subsection{Graphs and Laplacians}

For a graph $G$ and a subset of vertices $S$, let $G/S$ denote the graph obtained by \emph{identifying} $S$ to a single vertex $s$. Specifically, for any edge $e = \{u,v\}$ in $G$, replace each endpoint $u,v\in S$ with $s$ and do not change any endpoint not in $S$. Then, remove all self-loops to obtain $G/S$.

Let $G = (V(G),E(G))$ be a weighted undirected graph with $n$ vertices,
$m$ edges, and edge weights $\setof{w_e}_{e\in E(G)}$.
The Laplacian of $G$ is an $n\times n$ matrix given by:
\[
    (L_G)_{u,v} := \begin{cases}
        - w_{(u,v)} & \mathrm{if}\ u\neq v\ \mathrm{and}\ (u,v) \in E(G),\\
        \sum\nolimits_{(u,w) \in E(G)} w_{(u,w)} & \mathrm{if}\ u = v, \\
        0 & \mathrm{otherwise.}
    \end{cases}
\]
We define edge resistances $\setof{r_e}_{e\in E(G)}$ by $r_e = 1 / w_e$
for all $e\in E(G)$.

If we orient every edge $e \in E(G)$ arbitrarily,
we can define the signed edge-vertex incidence matrix $B_G$ by
\[
    (B_G)_{e,u} := \begin{cases}
        1 & \text{if $u$ is $e$'s head}, \\
        -1 & \text{if $u$ is $e$'s tail}, \\
        0 & \text{otherwise.}
    \end{cases}
\]
Then we can write $L_G$ as $L_G = B_G^T W_G B_G$,
where $W_G$ is a diagonal matrix with $(W_G)_{e,e} = w_e$.

For vertex sets $S,T\subseteq V$,
$(L_G)_{S,T}$ denotes the submatrix of $L_G$ with row indices in $S$
and column indices in $T$.

$L_G$ is always positive semidefinite, and only has one zero eigenvalue
if $G$ is connected. For a connected graph $G$,
let $0 = \lambda_1(L_G) < \lambda_2(L_G) \leq \ldots \leq \lambda_n(L_G)$
be the eigenvalues of $L_G$.
Let $u_1,u_2,\ldots,u_n$ be the corresponding set of orthonormal eigenvectors.
Then, we can diagonalize $L_G$ and write
\begin{align*}
    L_G = \sum_{i=2}^n \lambda_i(L_G) u_i u_i^T.
\end{align*}
The pseudoinverse of $L_G$ is then given by
\begin{align*}
    L_G^+ = \sum_{i=2}^n \frac{1}{\lambda_i(L_G)} u_i u_i^T.
\end{align*}

In the rest of the paper, we will write $\lambdamin(\cdot)$ to
denote the smallest eigenvalue and $\lambdamax(\cdot)$ to denote the largest
eigenvalue.
We will also write $\sigmamax(\cdot)$ to denote the largest singular value,
which is given by
\begin{align*}
    \sigmamax(A) = \sqrt{\lambdamax(A^T A)}
\end{align*}
for any matrix $A$.

We will also need to use Schur complements
which are defined as follows:

\begin{definition}[Schur Complements]\label{def:schur}
    The Schur complement of a graph $G$ onto a subset of vertices
    $S\subset V(G)$, denoted by $SC(G,S)$ or $SC(L_G,S)$,
    is defined as
    \[
        SC(L_G,S) = (L_G)_{S,S} - (L_G)_{S,T} (L_G)_{T,T}^{-1} (L_G)_{T,S},
    \]
    where $T := V(G) \setminus S$.
\end{definition}

The fact below relates Schur complements to the inverse of graph Laplacian:
\begin{fact}[see, e.g., Fact 5.4 in~\cite{DKPRS17}]\label{fact:schurinv}
    For any graph $G$ and $S\subset V(G)$,
    \[
        \kh{I - \frac{1}{\sizeof{S}} J} \kh{L_G^+}_{S,S} \kh{I - \frac{1}{\sizeof{S}} J}
        = \kh{SC(L_G, S)}^+,
    \]
    where $I$ denotes the identity matrix, and $J$ denotes the matrix
    whose entries are all $1$.
\end{fact}

\subsection{Leverage scores and rank one updates}

For a graph $G$ and an edge $e\in E(G)$, let $b_e\in \mathbb{R}^{V(G)}$ denote the signed indicator vector of the edge $e$; that is the vector with $-1$ on one endpoint, 1 on the other, and 0 everywhere else. Define the \emph{leverage score} of $e$ to be the quantity $$\texttt{lev}_G(e) := \frac{b_e^T L_G^+ b_e}{r_e}$$ Let $d_1,d_2\in \mathbb{R}^n$ be two vectors with $d_1,d_2\perp \textbf{1}^{V(G)}$. Then the following results hold by the Sherman-Morrison rank 1 update formula:

\begin{proposition}
For a graph $G$ and an edge $f$, let $G\backslash f$ denote the graph with $f$ deleted. Then $$d_1^T L_{G\backslash f}^+ d_2 = d_1^T L_G^+ d_2 + \frac{(d_1^T L_G^+ b_f)(b_f^T L_G^+ d_2)}{r_f - b_f^T L_G^+ b_f}$$
\end{proposition}

\begin{proposition}
For a graph $G$ and an edge $f$, let $G/f$ denote the graph with $f$ contracted. Then $$d_1^T L_{G/f}^+ d_2 = d_1^T L_G^+ d_2 - \frac{(d_1^T L_G^+ b_f)(b_f^T L_G^+ d_2)}{b_f^T L_G^+ b_f}$$
\end{proposition}

\subsection{Random spanning trees}

We use the following result on uniform random spanning tree generation:

\begin{theorem}[Theorem 1.2 of \cite{Sc17}]
Given a weighted graph $G$ with $m$ edges, a random spanning tree $T$ of $G$ can be sampled from a distribution with total variation distance at most $1/m^{10}$ from the uniform distribution in time $m^{1 + o(1)}$.
\end{theorem}

Let $T\sim G$ denote the uniform distribution over spanning trees of $G$. We also use the following classic result:

\begin{theorem}[\cite{K47}]
For any edge $e\in E(G)$, $\Pr_{T\sim G}[e\in T] = \texttt{lev}_G(e)$.
\end{theorem}

For an edge $e\in E(G)$, let $G[e]$ denote a random graph obtained by contracting $e$ with probability $\texttt{lev}_G(e)$ and deleting $e$ otherwise.

\subsection{Some useful bounds and tools}

We now describe
some useful bounds/tools we will need
in our algorithms.
In all the following bounds, we define
the quantities $\wmax$ and $\wmin$ as follows:
\begin{align*}
    \wmax &:= \max\setof{1, \max\nolimits_{e\in E(G)} 1/r_e}, \\
    \wmin &:= \min\setof{1, \min\nolimits_{e\in E(G)} 1/r_e}.
\end{align*}
%where $r_e$ denotes the resistance of edge $e$
%(the reciprocal edge conductance/weight).

The following lemma bounds
the range of eigenvalues for
Laplacians and SDDM matrices:

\begin{restatable}[]{lemma}{lemlambdas} \label{lem:lambdas}
    For any Laplacian $L_G$ and $S\subset V(G)$,
    \begin{align}
        \lambda_2(L_G)& \geq \wmin / n^2, \label{eq:eigen1} \\
        \lambdamin\kh{(L_G)_{S,S}} &\geq \wmin / n^2, \label{eq:eigen2} \\
        \lambdamax\kh{(L_G)_{S,S}} &\leq \lambdamax(L_G) \leq n \wmax. \label{eq:eigen3}
    \end{align}
\end{restatable}
\begin{proof}
    Defered to Appendix~\ref{sec:eigen}.
\end{proof}

The lemma below gives upper bounds on the largest
eigenvalues/singular values for some useful matrices:
\begin{restatable}[]{lemma}{lemsigmas} \label{lem:sigmas}
    The following upper bounds on the largest singular values/eigenvalues hold:
    \begin{align}
        &\sigmamax(W_G^{1/2} B_G) \leq (n\wmax)^{1/2}, \label{eq:singular1} \\
        &\lambdamax(SC(L_G,S)) \leq n \wmax, \label{eq:singular2} \\
        &\sigmamax((L_G)_{S,T}) = \sigmamax((L_G)_{T,S}) \leq n \wmax,
        \label{eq:singular3}
    \end{align}
    where $T := V(G) \setminus S$.
    %and $\sigmamax(\cdot)$ denotes the largest singular value of a matrix.
\end{restatable}
\begin{proof}
    Defered to Appendix~\ref{sec:2norm}.
\end{proof}

We will need to invoke Fast Laplacian Solvers
to apply the inverse of a Laplacian of an SDDM matrix.
The following lemma characterizes the performance
of Fast Laplacian Solvers:
\begin{lemma}[Fast Laplacian Solver~\cite{ST14,CKMPPRX14}]\label{lem:solve}
	There is an algorithm
    $\xtil = \LaplSolve(M, b, \ep)$
    %$\xx = \SDDMSolver(\SS, \bb, \eps)$
	which takes a matrix
	$M_{n\times n}$
    either a Laplacian or an SDDM matrix
	with $m$ nonzero entries,
	a vector $b \in \mathbb{R}^n$, and
	an error parameter $\ep > 0$,
	and returns a vector $\xtil \in \mathbb{R}^n$ such that
	%with high probability the following statement holds:
	\[
        \norm{x - \xtil }_{M} \leq \ep \norm{x}_M
    \]
	holds with high probability,
	where $\norm{x}_{M} := \sqrt{x^T M x}$,
    $x := M^{-1} b$,
	and $M^{-1}$ denotes the pseudoinverse of $M$
	when $M$ is a Laplacian.
	The algorithm runs in time
	$O(m\text{polylog}(n)\log(1/\ep))$.
\end{lemma}

The following lemmas show how to bound the errors of Fast Laplacian Solvers
in terms of $\ell_2$ norms, which follows directly from the bounds on
Laplacian eigenvalues in Lemma~\ref{lem:lambdas}:
\begin{restatable}[]{lemma}{lemlnorm} \label{lem:lnorm}
    For any Laplacian $L_G$, vectors $x,\xtil\in \mathbb{R}^n$
    both orthogal to $1$, and real number $\ep > 0$ satifiying
    \begin{align*}
        \norm{x - \xtil}_{L_G} \leq \ep \norm{x}_{L_G},
    \end{align*}
    the following statement holds:
    \begin{align*}
        \norm{x - \xtil} \leq \ep n^{1.5} \kh{\frac{\wmax}{\wmin}}^{1/2} \norm{x}.
    \end{align*}
\end{restatable}
\begin{proof}
    Defered to Appendix~\ref{sec:l2norm}.
\end{proof}

\begin{restatable}[]{lemma}{lemmnorm} \label{lem:mnorm}
    For any Laplacian $L_G$, $S\subset V$, vectors $x,\xtil\in\mathbb{R}^{\sizeof{S}}$,
    and real number $\ep > 0$ satifiying
    \begin{align*}
        \norm{x - \xtil}_{M} \leq \ep \norm{x}_{M},
    \end{align*}
    where $M := (L_G)_{S,S}$,
    the following statement holds:
    \begin{align*}
        \norm{x - \xtil} \leq \ep n^{1.5} \kh{\frac{\wmax}{\wmin}}^{1/2} \norm{x}.
    \end{align*}
\end{restatable}
\begin{proof}
    Defered to Appendix~\ref{sec:l2norm}.
\end{proof}

When computing the changes in effective resistances
due to the identification of a given vertex set
(i.e. merging vertices in that set and deleting any self loops formed),
we will need to use Johnson-Lindenstrauss
lemma to reduce dimensions:
\begin{lemma}[Johnson-Lindenstrauss Lemma~\cite{JL84,A01}]
	\label{lem:jl}
	Let $v_1,v_2,\ldots,v_n \in \mathbb{R}^d$
	be fixed vectors
	and $0 < \ep < 1$ be a real number.
	Let $k$ be a positive integer such that
	$
		k \geq 24 \log n / \ep^2
	$
    and $Q_{k \times d}$ be a random $\pm 1$ matrix.
	With high probability, the following statement holds
	for any $1\leq i, j\leq n$:
    \begin{align*}
        (1 - \ep) \norm{v_i - v_j}^2
        \leq \norm{Q v_i - Q v_j}^2 \leq
        (1 + \ep) \norm{v_i - v_j}^2.
    \end{align*}
\end{lemma}

\section{Existence of sparsifiers}

In this section, we reduce the construction of spectral subspace sparsifiers to an oracle that outputs edges that have low energy with respect to every demand vector in the chosen subspace $\mc S$. We prove it by splitting and conditioning on edges being present in a uniformly random spanning tree one-by-one until $\tilde{O}(d/\ep^2)$ edges are left. This construction is a high-dimensional generalization of the construction given in Section 10.1 of \cite{Sc17}. We use the following matrix concentration inequality:

\begin{theorem}[Matrix Freedman Inequality applied to symmetric matrices \cite{T11}]\label{thm:freedman}
Consider a matrix martingale $(Y_k)_{k\ge 0}$ whose values are symmetric matrices with dimension $s$, and let $(X_k)_{k\ge 1}$ be the difference sequence $X_k := Y_k - Y_{k-1}$. Assume that the difference sequence is uniformly bounded in the sense that

$$\lambda_{\max}(X_k)\le R$$

almost surely for $k \ge 1$. Define the predictable quadratic variation process of the martingale:

$$W_k := \sum_{j=1}^k \textbf{E}[X_j^2|Y_{j-1}]$$

Then, for all $t\ge 0$ and $\sigma^2 > 0$,

$$\Pr[\exists k\ge 0: \lambda_{\max}(Y_k - Y_0)\ge t \text{ and } \lambda_{\max}(W_k)\le \sigma^2] \le s \exp\left(\frac{-t^2/2}{\sigma^2 + Rt/3}\right)$$
\end{theorem}

Now, we give an algorithm $\SSS(G,\mc S,\ep)$ that proves Theorem \ref{thm:main-existence}. The algorithm simply splits and conditions on the edge that minimizes the martingale difference repeatedly until there are too few edges left. For efficiency purposes, $\SSS(G,\mc S,\ep)$ receives martingale-difference-minimizing edges from a steady oracle $\Oracle$ with the additional guarantee that differences remain small after many edge updates. This oracle is similar to the stable oracles given in Section 10 of \cite{Sc17}. 

\begin{definition}[Steady oracles]
A $(\rho,K(z))$-\emph{steady oracle} is a function $Z\gets \Oracle(I,\mc S)$ that takes in a graph $I$ and a subspace $\mc S\subseteq \mathbb{R}^{V(I)}$ that satisfy the following condition:

\begin{itemize}
\item (Leverage scores) For all $e\in E(I)$, $\texttt{lev}_I(e)\in [3/16,13/16]$.
\end{itemize}

and outputs a set $Z\subseteq E(I)$. Let $I_0 = I$ and for each $i > 0$, obtain $I_i$ by picking a uniformly random edge $f_{i-1}\in Z$, arbitrarily letting $I_i\gets I_{i-1}\setminus f_{i-1}$ or $I_i\gets I_{i-1}/f_{i-1}$, and letting $Z\leftarrow Z\setminus \{f_{i-1}\}$. $\Oracle$ satisfies the following guarantees with high probability for all $i < K(|E(I)|)$:

\begin{itemize}
\item (Size of $Z$) $|Z|\ge |E(I)|/\rho$
\item (Leverage score stability) $\texttt{lev}_{I_i}(f_i)\in [1/8,7/8]$
\item (Martingale change stability) $\max_{x\in \mc S} \frac{(x_{I_i}^T L_{I_i}^+ b_{f_i})^2}{r_{f_i} (x^T L_I^+ x)} \le \frac{\rho \Dim(\mc S)}{|E(I)|}$
\end{itemize}
\end{definition}

We now state the main result of this section:

\begin{lemma}\label{lem:main-oracle}
Consider a weighted graph $G$, a $d$-dimensional vector space $\mc S \subseteq \mathbb{R}^{V(G)}$, and $\ep\in (0,1)$. There is an algorithm $\SSS(G,\mc S,\ep)$ that, given access to a $(\rho,K(z))$-steady-oracle $\Oracle$, computes a $(\mc S,\ep)$-spectral subspace sparsifier for $G$ with $$O\left(\frac{\rho^2 d \log d}{\ep^2}\right)$$ edges in time $$O\left((\log n) (\max_{z\le |E(G)|} z/K(z))(\mc T_{rst} + \mc T_{\Oracle} + m)\right)\le O\left((\log n) (\max_{z\le |E(G)|} z/K(z)) (m^{1+o(1)} + \mc T_{\Oracle})\right)$$ where $\mc T_{rst}$ is the time required to generate a spanning tree of $G$ from a distribution with total variation distance $\le n^{-10}$ from uniform and $\mc T_{\Oracle}$ is the runtime of the oracle.
\end{lemma}

The algorithm will use two simple subroutines that modify the graph by splitting edges. $\Split$ replaces each edge with approximate leverage score less than 1/2 with a two-edge path and each edge with approximate leverage score greater than 1/2 with two parallel edges. $\Unsplit$ reverses this split for all pairs that remain in the graph. We prove the following two results about this subroutines in the appendix:

\begin{restatable}{proposition}{propsplit}\label{prop:split}
There is a linear-time algorithm $(I,\mc P)\gets \Split(H)$ that, given a graph $H$, produces a graph $I$ with $V(H)\subseteq V(I)$ and a set of pairs of edges $\mc P$ with the following additional guarantees:

\begin{itemize}
\item (Electrical equivalence) For all $x\in \mathbb{R}^{V(I)}$ that are supported on $V(H)$, $x^T L_I^+ x = x_H^T L_H^+ x_H$.
\item (Bounded leverage scores) For all $e\in E(I)$, $\texttt{lev}_I(e)\in [3/16,13/16]$
\item ($\mc P$ description) Every edge in $I$ is in exactly one pair in $\mc P$. Furthermore, there is a bijection between pairs $(e_0,e_1)\in \mc P$ and edges $e\in E(H)$ for which either (a) $e_0,e_1$ and $e$ have the same endpoint pair or (b) $e_0 = \{u,w\}$, $e_1 = \{w,v\}$, and $e = \{u,w\}$ for some degree 2 vertex $w$.
\end{itemize}
\end{restatable}

\begin{restatable}{proposition}{propunsplit}\label{prop:unsplit}

There is a linear-time algorithm $H\gets \Unsplit(I,\mc P)$ that, given a graph $I$ and a set of pairs $\mc P$ of edges in $I$, produces a minor $H$ with $V(H)\subseteq V(I)$ and the following additional guarantees:

\begin{itemize}
\item (Electrical equivalence) For all $x\in \mathbb{R}^{V(I)}$ that are supported on $V(H)$, $x^T L_I^+ x = x_H^T L_H^+ x_H$.
\item (Edges of $H$) There is a surjective map $\phi:E(I)\rightarrow E(H)$ from non-self-loop,non-leaf edges of $I$ such that for any pair $(e_0,e_1)\in \mc P$, $\phi(e_0) = \phi(e_1)$. Furthermore, for each $e\in E(H)$, either (a) $\phi^{-1}(e) = e$, (b) $\phi^{-1}(e) = \{e_0,e_1\}$, with $(e_0,e_1)\in \mc P$ and $e_0,e_1$ having the same endpoints as $e$ or (c) $\phi^{-1}(e) = \{e_0,e_1\}$, with $(e_0,e_1)\in \mc P$ and $e_0 = \{u,w\},e_1 = \{w,v\}$, and $e = \{u,v\}$ for a degree 2 vertex $w$.
\end{itemize}
\end{restatable}

\begin{algorithm}[H]
\SetAlgoLined
\DontPrintSemicolon

    \KwIn{A weighted graph $G$, a vector space $\mc S\subseteq \mathbb{R}^{V(G)}$, $\ep\in (0,1)$, and (implicitly) a $(\rho,K(z))$-steady oracle $\Oracle$}

    \KwOut{A $(\mc S,\ep)$-spectral subspace sparsifier for $G$}

    $H\gets G$\;

    \While{$|E(H)|\ge 10000\rho^2(\Dim(\mc S) \log (\Dim(\mc S)))/\ep^2$}{

        $(I,\mc P)\gets \Split(H)$\;

        $Z\gets \Oracle(I,\mc S)$\;

        $Z'\gets $ a uniformly random subset of $Z$ with size $K(|E(I)|)$\; \label{line:slow-e-choice}

        $T\gets $ a spanning tree of $I$ drawn from a distribution with TV distance $\le \kappa_0 := 1/n^{10}$ from uniform\;

        $I'\gets $ the graph with all edges in $E(T)\cap Z'$ contracted and all edges in $Z'\setminus E(T)$ deleted\;

        $H\gets \Unsplit(I',\mc P)$\;

    }

    \Return{$H$}

\caption{$\SSS(G,\mc S,\ep)$}
\end{algorithm}

We analyze the approximation guarantees of $H$ by setting up two families of matrix-valued martingales. In all of the proof besides the final ``Proof of Lemma \ref{lem:main-oracle},'' we sample $T$ from the uniform distribution rather than from a distribution with total variation distance $\kappa_0$ from uniform. We bound the error incurred from doing this in the final ``Proof of Lemma \ref{lem:main-oracle}.''

We start by defining the first family, which just consists of one martingale. Let $H_0 := G$ and let $H_k$ be the graph $H$ between iterations $k$ and $k+1$ of the while loop of $\SSS$. Let $d = \Dim(\mc S)$. Since $\mc S$ is orthogonal to $\textbf{1}^{V(G)}$, $\Dim((L_G^+)^{1/2} \mc S) = \Dim(\mc S) = d$, which means that $\mc S$ has a basis $\{y_i\}_{i=1}^d$ for which $y_i^T L_G^+ y_j = 0$ for all $i\ne j\in [d]$ and $y_i^T L_G^+ y_i = 1$ for all $i\in [d]$. Let $Y_k$ be the $|V(H_k)|\times d$ matrix with $i$th column $(y_i)_{H_k}$ and let $Y := Y_0$. Let $M_k := Y_k^TL_{H_k}^+Y_k - Y^TL_G^+Y$. Since the $y_i$s form a basis of $\mc S$, there is a vector $a_x$ for which $x = Ya_x$ for any $x\in \mc S$. Furthermore, $x_{H_k} = Y_ka_x$ for any $k\ge 0$. In particular, $$\left|\frac{x_{H_k}^T L_{H_k}^+ x_{H_k}}{x^T L_G^+ x} - 1\right| = \left|\frac{a_x^T M_k a_x}{||a_x||_2^2}\right|$$ so it suffices to show that $\lambda_{\max}(M_k)\le \ep$ for all $k \le k_{\text{final}}$, where $k_{\text{final}}$ is the number of while loop iterations.

In order to bound the change between $M_k$ and $M_{k+1}$, we introduce a second family of martingales consisting of one martingale for each while loop iteration. Let $I_{k,0} := I$ during the $k$th iteration of the while loop in $\SSS$. Generate $Z'$ in $Z$ during iteration $k$ of the while loop by sampling a sequence of edges $f_{k,0},f_{k,1},\hdots,f_{k,K(|E(I)|)-1}$ without replacement from $Z$. Let $I_{k,t} = I_{k,t-1}[f_{k,t-1}]$ for all $t > 0$. For a vector $v\in \mathbb{R}^{V(G)}$, let $v_{I_{k,0}}\in \mathbb{R}^{V(I_{k,0})}$ be the vector with $v_{I_{k,0}}(p) = v_{H_k}(p)$ for $p\in V(H_k)$ and $v_{I_{k,0}}(p) = 0$ for $p\in V(I_{k,0})\setminus V(H_k)$. For $t > 0$ and $v\in \mathbb{R}^{V(G)}$, let $v_{I_{k,t}} := (v_{I_{k,0}})_{I_{k,t}}$. Let $Y_{k,t}$ be the $|V(I_{k,t})|\times d$ matrix with $i$th column $(y_i)_{I_{k,t}}$. Let $N_{k,t} := Y_{k,t}^TL_{I_{k,t}}^+ Y_{k,t} - Y^TL_G^+Y$. For any $x\in \mc S$, $t\ge 0$, and $k\ge 0$, $x_{I_{k,t}} = Y_{k,t}a_x$. In particular, $$\left|\frac{x_{I_{k,t}}^T L_{I_{k,t}}^+ x_{I_{k,t}}}{x^T L_G^+ x} - 1\right| = \left|\frac{a_x^T N_{k,t} a_x}{||a_x||_2^2}\right|$$
Next, we write an equivalent formulation for the steady oracle ``Martingale change stability'' guarantee that is easier to analyze:

\begin{proposition}\label{prop:equiv}
$$\max_{x\in \mc S} \frac{(x_{I_{k,t}}^T L_{I_{k,t}}^+ b_f)^2}{r_f(x^T L_G^+ x)} = \frac{b_f^T L_{I_{k,t}}^+ Y_{k,t} Y_{k,t}^T L_{I_{k,t}}^+ b_f}{r_f}$$
\end{proposition}

\begin{proof}
Notice that

\begin{align*}
\max_{x\in \mc S} \frac{(x_{I_{k,t}}^T L_{I_{k,t}}^+ b_f)^2}{r_f(x^T L_G^+ x)} &= \max_{x\in \mc S} \frac{(a_x^T Y_{k,t}^T L_{I_{k,t}}^+ b_f) (b_f^T L_{I_{k,t}}^+ Y_{k,t} a_x)}{r_f||a_x||_2^2}\\
&= \max_{a\in \mathbb{R}^d} \frac{a^T Y_{k,t}^T L_{I_{k,t}}^+ b_f b_f^T L_{I_{k,t}}^+ Y_{k,t} a}{r_f||a||_2^2}\\
&= \lambda_{\max}\left(\frac{Y_{k,t}^T L_{I_{k,t}}^+ b_f b_f^T L_{I_{k,t}}^+ Y_{k,t}}{r_f}\right)\\
&= \frac{b_f^T L_{I_{k,t}}^+ Y_{k,t} Y_{k,t}^T L_{I_{k,t}}^+ b_f}{r_f}\\
\end{align*}

\noindent as desired.
\end{proof}

Now, we analyze the inner family of matrices $N_{k,t}$. Let $Z_{k,t}$ denote the set $Z$ during iteration $k$ of the while loop after sampling $t$ edges without replacement.

\begin{proposition}\label{prop:martingale}
$Y_t := N_{k,t}$ for fixed $k\ge 0$ and varying $t\ge 0$ is a matrix martingale. Furthermore, if
$$\frac{x_{I_{k,s}}^T L_{I_{k,s}}^+ x_{I_{k,s}}}{x^T L_G^+ x}\le 10$$
for all $x\in \mc S$, $k\ge 0$, and $s\le t$ for some $t\ge 0$, $\lambda_{\max}(X_{t+1})\le \frac{90d}{|E(I_{k,t})|}$ and $\lambda_{\max}(\textbf{E}[X_{t+1}^2|Y_t])\le \frac{25600\rho^2 d}{|E(I_{k,0})|^2}$, where $X_{t+1}$ is defined based on the $Y_s$s as described in Theorem \ref{thm:freedman}.
\end{proposition}

\begin{proof}
We compute the conditional expectation of $X_{t+1} = Y_{t+1} - Y_t$ given $Y_t$ using Sherman-Morrison:

\begin{align*}
\textbf{E}[X_{t+1} | Y_t] &= \textbf{E}[N_{k,t+1} - N_{k,t} | N_{k,t}]\\
&= \frac{1}{|Z_{k,t}|}\sum_{f\in Z_{k,t}} -\frac{b_f^T L_{I_{k,t}}^+ b_f}{r_f}\left(\frac{Y_{k,t}^T L_{I_{k,t}}^+ b_f b_f^T L_{I_{k,t}}^+ Y_{k,t}}{b_f^T L_{I_{k,t}}^+ b_f}\right)\\
&+ \frac{1}{|Z_{k,t}|}\sum_{f\in Z_{k,t}} \left(1 - \frac{b_f^T L_{I_{k,t}}^+ b_f}{r_f}\right)\left(\frac{Y_{k,t}^T L_{I_{k,t}}^+ b_f b_f^T L_{I_{k,t}}^+ Y_{k,t}}{r_f - b_f^T L_{I_{k,t}}^+ b_f}\right)\\
&= 0\\
\end{align*}

Therefore, $(Y_t)_{t\ge 0}$ is a martingale. Since $I_{k,0}$ is the output of $\Split$, all edges in $I_{k,0}$ have leverage score between $3/16$ and $13/16$ by Proposition \ref{prop:split}. In particular, the input condition to $\Oracle$ is satisfied. Furthermore,

\begin{align*}
\lambda_{\max}(X_{t+1}) &\le \lambda_{\max}(N_{k,t+1} - N_{k,t})\\
&\le \frac{1}{|Z_{k,t}|} \sum_{f\in Z_{k,t}} \lambda_{\max}\left(\frac{Y_{k,t}^T L_{I_{k,t}}^+ b_f b_f^T L_{I_{k,t}}^+ Y_{k,t}}{\min(b_f^T L_{I_{k,t}}^+ b_f, r_f - b_f^T L_{I_{k,t}}^+ b_f)}\right)\\
&\le 8\max_{f\in Z_{k,t}} \lambda_{\max}\left(\frac{Y_{k,t}^T L_{I_{k,t}}^+ b_f b_f^T L_{I_{k,t}}^+ Y_{k,t}}{r_f}\right)\\
&= 8\max_{f\in Z_{k,t}} \max_{x\in \mc S}\frac{(x_{I_{k,t}}^T L_{I_{k,t}}^+ b_f)^2}{r_f (x^T L_G^+ x)}\\
&\le 80\max_{f\in Z_{k,t}} \max_{x\in \mc S}\frac{(x_{I_{k,t}}^T L_{I_{k,t}}^+ b_f)^2}{r_f (x_{I_{k,0}}^T L_{I_{k,0}}^+ x_{I_{k,0}})}\\
&\le \frac{90\rho d}{|E(I_{k,0})|}\\
\end{align*}

where the third inequality follows from ``Leverage score stability,'' the equality follows from Proposition \ref{prop:equiv}, the fourth inequality follows from the input condition, and the last inequality follows from ``Martingale change stability.'' Also, 

\begin{align*}
\lambda_{\max}(\textbf{E}[X_{t+1}^2|Y_t]) &= \lambda_{\max}\left(\textbf{E}[(N_{k,t+1} - N_{k,t})^2|Y_t]\right)\\
&\le \frac{256}{|Z_{k,t}|}\lambda_{\max}\left(\sum_{f\in Z_{k,t}} \frac{1}{r_f^2} Y_{k,t}^T L_{I_{k,t}}^+ b_f b_f^T L_{I_{k,t}}^+ Y_{k,t}Y_{k,t}^T L_{I_{k,t}}^+ b_f b_f^T L_{I_{k,t}}^+ Y_{k,t} \right)\\
&\le \frac{256d}{|Z_{k,t}|}\left(\max_{f\in Z_{k,t}} \frac{1}{r_f}b_f^T L_{I_{k,t}}^+ Y_{k,t} Y_{k,t}^T L_{I_{k,t}}^+ b_f\right)\lambda_{\max}\left(\sum_{f\in Z_{k,t}} \frac{1}{r_f} Y_{k,t}^T L_{I_{k,t}}^+ b_f b_f^T L_{I_{k,t}}^+ Y_{k,t} \right)\\
&\le \frac{2560\rho d}{|Z_{k,t}||E(I_{k,0})|}\lambda_{\max}\left(Y_{k,t}^T L_{I_{k,t}}^+ Y_{k,t} \right)\\
&= \frac{2560\rho d}{|Z_{k,t}||E(I_{k,0})|}\max_{x\in \mc S} \frac{x_{I_{k,t}}^T L_{I_{k,t}}^+ x_{I_{k,t}}}{x^T L_G^+ x}\\
&\le \frac{25600\rho^2 d}{|E(I_{k,0})|^2}\\
\end{align*}

where the second inequality follows from Sherman-Morrison and ``Leverage score stability,'' the fourth follows from ``Martingale change stability,'' and the last follows from ``Size of $Z$'' and the input condition.
\end{proof}

Now, consider the sequence of matrices $((N_{k,t})_{t=0}^{K(|E(I_{k,t})|)})_{k\ge 0}$ obtained by concatenating the $(N_{k,t})_t$ martingales for each $k$. We now analyze this sequence:

\begin{proposition}\label{prop:big-martingale}
The sequence of matrices $(Y_{kt})_{k,t}$ ordered lexicographically by $(k,t)$ pairs defined by $Y_{kt} := N_{k,t}$ is a matrix martingale. Furthermore, if for any $k\ge 0,t\ge 0$, any pairs $(l,s)$ lexicographically smaller than $(k,t)$, and any $x\in \mc S$,
$$\frac{x_{I_{l,s}}^T L_{I_{l,s}}^+ x_{I_{l,s}}}{x^T L_G^+ x}\le 10$$
then $$\lambda_{\max}(X_{k't'})\le \frac{90d}{|E(I_{k',t'})|}$$, $\lambda_{\max}(\textbf{E}[X_{k't'}^2|Y_{kt}])\le \frac{25600\rho^2 d}{|E(I_{k',t'})|^2}$, and $$\lambda_{\max}(W_{kt})\le \sum_{(l,s) \le (k,t)} \frac{25600\rho^2 d}{|E(I_{l,s})|^2}$$ where $(k',t')$ is the lexicographic successor to $(k,t)$ and $X_{k't'} = Y_{k't'} - Y_{kt}$ as described in Theorem \ref{thm:freedman}.
\end{proposition}

\begin{proof}
Consider a pair $(k',t')$ with lexicographic predecessor $(k,t)$. If $k = k'$, then $t' = t+1$, which means that $$\textbf{E}[Y_{k't'}|Y_{kt}] = \textbf{E}[N_{k,t+1}|N_{k,t}] = N_{k,t} = Y_{kt}$$, $\lambda_{\max}(X_{k't'})\le \frac{90d}{|E(I_{k',t'})|}$, and $\lambda_{\max}(\textbf{E}[X_{k't'}^2|Y_{kt}]) \le \frac{25600\rho^2 d}{|E(I_{k',t'})|^2}$ by Proposition \ref{prop:martingale}. If $k = k'-1$, then $t' = 0$. As a result, $Y_{kt} = N_{k,t} = M_k$ by the ``Electrical equivalence'' guarantee of Proposition \ref{prop:unsplit} and $M_k = N_{k',t'} = Y_{k't'}$ by the ``Electrical equivalence'' guarantee of Proposition \ref{prop:split}. In particular, $X_{k't'} = 0$ and satisfies the desired eigenvalue bounds. The bound on $\lambda_{\max}(W_{kt})$ follows directly from the stepwise bound $\lambda_{\max}(\textbf{E}[X_{k't'}^2|Y_{kt}])$ and the definition of $W_{kt}$ as the predictable quadratic variation.
\end{proof}

Now, we are ready to prove Lemma \ref{lem:main-oracle}.

\begin{proof}[Proof of Lemma \ref{lem:main-oracle}]

\textbf{$H$ a minor of $G$.} It suffices to show that for every $k\ge 0$, $H_{k+1}$ is a minor of $H_k$. $I'$ is a minor of $I$, as $I$ is only modified by deletion and contraction. Now, we show that the unweighted version of $H_{k+1}$ can be obtained from $H_k$ by contracting each edge $e\in E(H_k)$ with an $I'$-self-loop in its pair $(e_0,e_1)\in P$ and deleting each edge $e\in E(H_k)$ with an $I'$-leaf edge in its pair. Let $H_{k+1}'$ be the result of this procedure.

We show that $H_{k+1}' = H_{k+1}$ without weights. We start by showing that $V(H_{k+1}) = V(H_{k+1}')$. Each vertex $v\in V(H_{k+1})$ corresponds to a set of vertices in $V(H_k)$ that were identified, as the ``Edges of $H$'' requirement ensures that $H_{k+1}$ contains no vertices that were added to $H_k$ by $\Split$. Since $T$ is a tree, each vertex $v\in V(H_{k+1})$ corresponds to a subtree of identified vertices in $H_k$. Since $Z$ only contains one edge for each pair in $\mc P$, the self-loop edges in $I'$ match the edges contracted to form the subtree for $v$, which means that $V(H_{k+1}) = V(H_{k+1}')$. $E(H_{k+1}) \subseteq E(H_{k+1}')$ because for every $e\in E(H_{k+1})$, $\phi^{-1}(e)$ does not contain an $I'$ self-loop or leaf by the ``Edges of $H$'' and ``$\mc P$ description'' guarantees. $E(H_{k+1}') \subseteq E(H_{k+1})$ because each $e\in E(H_{k+1}')$ does not map to a self-loop or leaf in $I'$, which means that $\phi^{-1}(e)$ exists by surjectivity of $\phi$. Therefore, $H_{k+1}' = H_{k+1}$. Since $H_{k+1}'$ is a minor of $H_k$, $H_{k+1}$ is also a minor of $H_k$, as desired.

\textbf{Number of edges.} This follows immediately from the while loop termination condition.

\textbf{Approximation bound.} Let $(k_{\tau},t_{\tau})$ be the final martingale index pair that the while loop encounters before termination. We start by obtaining a high-probability bound on $W_{k_{\tau}t_{\tau}}$ given that $T$ is drawn from the exact uniform distribution on spanning trees of $I$. By Proposition \ref{prop:big-martingale}, $$W_{k_{\tau}t_{\tau}}\le \sum_{(k,t)\le (k_{\tau},t_{\tau})} \frac{25600\rho^2d}{|E(I_{k,t})|^2}$$ The process of generating $I_{k,t+1}$ from $I_{k,t}$ does not increase the number of edges and decreases the number of edges by 1 with probability at least 1/8, by ``Leverage score stability.'' Therefore, by Azuma's Inequality, $|E(I_{k,t})|\le 2|E(G)| - c_{k,t}/8 + 10\sqrt{\log d}\sqrt{c_{k,t}}$ with probability at least $1 - 1/d^5$, where $c_{k,t}$ is the number of pairs that are lexicographically less than $(k,t)$. Therefore, as long as $|E(G)| > 20\log d$, which is true when $d > 10000000 = \Theta(1)$, $$|E(I_{k,t})|\le 2|E(G)| - c_{k,t}/16$$ with probability at least $1 - 1/d^3$ for all pairs $(k,t)$. This means that $$c_{k_{\tau},t_{\tau}}\le 32000000|E(G)|$$ and that $$W_{k_{\tau},t_{\tau}}\le \frac{32000000000 \rho^2 d}{|E(I_{k_{\tau},t_{\tau}})|}\le \ep^2/(10\log d)$$ with probability at least $1 - 1/d^3$.

Now, we apply the Matrix Freedman Inequality (Theorem \ref{thm:freedman}). Apply it to the martingale $(Y_{kt})_{k,t}$ to bound $\lambda_{\max}(Y_{k_{\tau}t_{\tau}} - Y_{00})$. By Proposition \ref{prop:big-martingale} and the termination condition for the while loop, we may set $R\gets \ep/(10\log d)\ge \frac{90d}{|E(I_{k_{\tau},t_{\tau}})|}$. By Theorem \ref{thm:freedman}, $$\Pr\left[\lambda_{\max}(Y_{k_{\tau}t_{\tau}}) \ge \ep \text{ and } \lambda_{\max}(W_{k_{\tau}t_{\tau}})\le \ep^2/(10\log d)\right]\le d\exp\left(\frac{-\ep^2/2}{\ep^2/(10\log d) + \ep^2/(30\log d)}\right)\le 1/d^2$$ Therefore, $$\Pr_{T \text{ uniform}}\left[\lambda_{\max}(Y_{k_{\tau}t_{\tau}}) \ge \ep\right]\le 1/d^2 + 1/d^3\le 2/d^2$$ Now, switch uniform spanning tree sampling to $\kappa_0$-approximate random spanning tree sampling. The total number of iterations is at most $m$, so the total TV distance of the joint distribution sampled throughout all iterations is at most $m\kappa_0$. Therefore, $$\Pr_{T \text{ $\kappa_0$-uniform}}\left[\lambda_{\max}(Y_{k_{\tau}t_{\tau}})\right]\le 2/d^2 + m\kappa_0\le 3/d^2$$ In particular, with probability at least $1 - 3/d^2$, $$\left|\frac{x_{H_{k_{\tau}}} L_{H_{k_{\tau}}}^+ x_{H_{k_{\tau}}}}{x^T L_G^+ x} - 1\right| \le \lambda_{\max}(M_{k_{\tau}}) = \lambda_{\max}(Y_{k_{\tau}t_{\tau}})\le \ep$$ for all $x\in \mc S$, as desired.

\textbf{Runtime.} By Azuma's Inequality,
$$|E(H_k)|\le |E(H_{k-1})| - K(|E(H_{k-1})|)/32\le (1 - \min_{z\ge 0} K(z)/32z) |E(H_{k-1})|$$
for all $k\le k_{\tau}$ with probability at least $1 - 1/d^2$. Therefore,
$$|E(H_k)|\le (1 - \min_{z\ge 0} K(z)/(32z))^k |E(G)|$$
which means that the termination condition is satisfied with high probability after
$$O((\log n) \max_{z\le |E(G)|} z/K(z))$$
iterations with high probability. Each iteration samples one spanning tree, calls the oracle once, and does a linear amount of additional work, yielding the desired runtime.
\end{proof}

\subsection{Slow oracle and proof of existence}

In this section, we prove Theorem \ref{thm:main-existence} by exhibiting a $(2,1)$-steady oracle $\SlowOracle(I,\mc S)$. The oracle just returns all edges in the bottom half by maximum energy fraction:

\begin{algorithm}[H]
\SetAlgoLined
\DontPrintSemicolon

    \KwIn{A graph $I$ and a subspace $\mc S\subseteq \mathbb{R}^{V(I)}$}

    \KwOut{A set $Z$ of edges satisfying the steady oracle definition}

    \Return{ all edges $e\in E(I)$ with $\max_{x\in \mc S} \frac{(x^T L_I^+ b_e)^2}{r_e (x^T L_I^+ x)}\le \frac{2\Dim(\mc S)}{|E(I)|}$}\;

\caption{$\SlowOracle$, never executed}
\end{algorithm}

To lower bound the number of edges added to $Z$, we use the following result and Markov's Inequality:

\begin{proposition}\label{prop:sum-bound}
$$\sum_{f\in E(I)} \max_{x\in \mc S} \frac{(x^T L_I^+ b_f)^2}{r_f (x^T L_I^+ x)} = \Dim(\mc S)$$
\end{proposition}

\begin{proof}

Let $Y_I$ be a $V(I)\times \Dim(\mc S)$-matrix consisting of a basis $(y_i)_{i=1}^d$ for $\mc S$ with $y_i^T L_I^+ y_j = 0$ for all $i\ne j\in [\Dim(\mc S)]$ and $y_i^T L_I^+ y_i = 1$ for all $i\in [\Dim(\mc S)]$. By Proposition \ref{prop:equiv},

\begin{align*}
\sum_{f\in E(I)} \max_{x\in \mc S} \frac{(x^T L_I^+ b_f)^2}{r_f (x^T L_I^+ x)} &=  \sum_{f\in E(I)} \frac{b_f^T L_I^+ Y_I Y_I^T L_I^+ b_f}{r_f}\\
&= \sum_{f\in E(I)} \Tr\left(\frac{b_f^T L_I^+ Y Y^T L_I^+ b_f}{r_f}\right)\\
&= \sum_{f\in E(I)} \Tr\left(\frac{L_I^+ Y_I Y_I^T L_I^+ b_fb_f^T}{r_f}\right)\\
&= \Tr(L_I^+ Y_I Y_I^T)\\
&= \sum_{i=1}^d y_i^T L_I^+ y_i\\
&= \Dim(\mc S)\\
\end{align*}

as desired.
\end{proof}

Now, we prove Theorem \ref{thm:main-existence}:

\begin{proof}[Proof of Theorem \ref{thm:main-existence}]

By Lemma \ref{lem:main-oracle}, it suffices to show that $\SlowOracle$ is a $(2,1)$-steady oracle.

\textbf{Size of $Z$.} By Markov's Inequality and Proposition \ref{prop:sum-bound}, $|Z|\ge |E(I)|/2$.

\textbf{Leverage score stability.} We are only interested in $i = 0$, for which the ``Leverage score'' input condition immediately implies the ``Leverage score stability'' guarantee.

\textbf{Martingale change stability.} We are only interested in $i = 0$. The return statement specifies the ``Martingale change stability'' guarantee for $\rho = 2$.
\end{proof}

\section{Fast oracle}

In this section, we give a $(O(\log^3 n), \Omega(z/\log^3 n))$-steady oracle $\FastOracle$ that proves Theorem \ref{thm:schur-fast} when plugged into $\SSS$. To do this, we use localization \cite{SRS17} to find a set of edges whose leverage scores and martingale changes do not change much over time. We use sketching and Lemma \ref{lem:diff-apx} to find these edges efficiently. This section can be described using the \textit{flexible function} framework given in \cite{Sc17}, but we give a self-contained treatment here.

\subsection{Efficient identification of low-change edges}

$\FastOracle$ needs to find a large collection of edges whose electrical energies do not change over the course of many iterations. This collection exists by the following result:

\begin{theorem}[Theorem 1 of \cite{SRS17}]\label{thm:localization}
Let $I$ be a graph. Then for any vector $w\in \mathbb{R}^{E(I)}$, $$\sum_{e,f\in E(I)} w_ew_f \frac{|b_e^T L_I^+ b_f|}{\sqrt{r_e}\sqrt{r_f}} \le c_{\text{local}}(\log^2 n) ||w||_2^2$$ for some constant $c_{\text{local}}$.
\end{theorem}

\noindent Plugging in $w \gets \textbf{1}^{E(I)}$ shows that at least half of the edges $e\in E(I)$, $$\sum_{f\in E(I)} \frac{|b_e^T L_I^+ b_f|}{\sqrt{r_e}\sqrt{r_f}}\le 2c_{\text{local}}\log^2 n$$ We decrease this bound by subsampling the edges in $I$ to obtain $Z$. To identify the edges with low sum, we use matrix sketching:

\begin{theorem}[Theorem 3 of \cite{I06} stated for $\ell_1$]\label{thm:sketch}
An efficiently computable, $\text{polylog}(d)$-space linear sketch exists for $\ell_1$ norms. That is, given a $d\in \mathbb{Z}_{\ge 1}$, $\delta\in (0,1)$, and $\ep\in (0,1)$, there is a matrix $C = \SketchMatrix(d,\delta,\ep)\in \mathbb{R}^{l\times d}$ and an algorithm $\RecoverNorm(s,d,\delta,\ep)$ with the following properties:

\begin{itemize}
\item (Approximation) For any vector $v\in \mathbb{R}^d$, with probability at least $1 - \delta$ over the randomness of $\SketchMatrix$, the value $r = \RecoverNorm(Cv,d,\delta,\ep)$ is as follows: $$(1-\ep)||v||_1\le r\le (1+\ep)||v||_1$$
\item $l = c/\ep^2\log(1/\delta)$
\item (Runtime) $\SketchMatrix$ and $\RecoverNorm$ take $O(ld)$ and $\text{poly}(l)$ time respectively.
\end{itemize}
\end{theorem}

\subsubsection{Approximation of column norms}

Consider a graph $I$ and a set $W\subseteq E(I)$. We can obtain multiplicative approximations the quantities $\sum_{f\in W} \frac{|b_e^T L_I^+ b_f|}{\sqrt{r_e}\sqrt{r_f}}$ for all $e\in W$ in near-linear time using Theorem \ref{thm:sketch}. However, we actually need to multiplicatively approximate the quantities $\sum_{f\in W, f\ne e} \frac{|b_e^T L_I^+ b_f|}{\sqrt{r_e}\sqrt{r_f}}$. In particular, we need to estimate the $\ell_1$ norm of the rows of the matrix $M$ with $M_{ef} := \frac{|b_e^T L_I^+ b_f|}{\sqrt{r_e}\sqrt{r_f}}$ with the diagonal left out. To do this, we tile the matrix as described in Section 11.3.2 of \cite{Sc17}:

\begin{itemize}
\item Do $\Theta(\log n)$ times:
    \begin{itemize}
    \item Pick a random balanced partition $(W_0,W_1)$ of $W$
    \item For each $e\in W_0$, approximate $a_e\gets \sum_{f\in Z_1} \frac{|b_e^T L_I^+ b_f|}{\sqrt{r_e}\sqrt{r_f}}$ using sketching
    \end{itemize}
\item For each $e\in W$, average the $a_e$s together and scale up the average by a factor of 4 to obtain an estimate for $\sum_{f\ne e\in W} \frac{|b_e^T L_I^+ b_f|}{\sqrt{r_e}\sqrt{r_f}}$
\end{itemize}

The expected contribution of each off-diagonal entry is 1, while no diagonal entry can contribute. After $\Theta(\log n)$ trials, the averages concentrate by Chernoff and a union bound. Now, we formally implement this idea:

\begin{proposition}\label{prop:estimate-offdiag}
There is a near-linear time algorithm $(a_e)_{e\in W}\gets \ColumnApx(I,W)$ that takes a graph $I$ and a set of edges $W\subseteq E(I)$ and returns estimates $a_e$ for which $$a_e/2\le \sum_{f\ne e\in W} \frac{|b_e^T L_I^+ b_f|}{\sqrt{r_e}\sqrt{r_f}}\le 3a_e/2$$ for all $e\in W$.
\end{proposition}

\begin{algorithm}[H]
\SetAlgoLined
\DontPrintSemicolon
\caption{$\ColumnApx(I,W)$}

    \KwIn{a graph $I$ and $W\subseteq E(I)$}

    \KwOut{approximations to the values $\{\sum_{f\ne e\in W} \frac{|b_e^T L_I^+ b_f|}{\sqrt{r_e}\sqrt{r_f}}\}_{e\in W}$}

    $K\gets 100\log n$\;

    $\kappa_e\gets 0$ for each $e\in W$\;

    \ForEach{$k\gets 1,2,\hdots,K$}{

        $W_0,W_1\gets $ partition of $W$ with $e\in W_0$ or $e\in W_1$ i.i.d. with probability 1/2\;

        $C\gets \SketchMatrix(|W_1|,1/n^6,1/4)$\; \label{line:sketch-mat}

        $D\gets V(I)\times |W_1|$ matrix of columns $b_f/\sqrt{r_f}$ for $f\in W_1$\;

        $U\gets L_I^+ DC^T$\; \label{line:compute-u}

        \ForEach{$e\in W_0$}{

            Increase $\kappa_e$ by $\RecoverNorm(U^T(b_e/\sqrt{r_e}),|W_1|,1/n^6,1/4)$\; \label{line:recover-norm}

        }

    }

    \Return{$(4\kappa_e/K)_{e\in W}$}
\end{algorithm}

\begin{proof}[Proof of Proposition \ref{prop:estimate-offdiag}]

\textbf{Approximation.} Let $Y_{ef}^{(k)}$ be the indicator variable of the event $\{e\in W_0 \text{ and } f\in W_1 \text{ in iteration $k$}\}$. By the ``Approximation'' guarantee of Theorem \ref{thm:sketch}, at the end of the foreach loop in $\ColumnApx$,

$$\kappa_e\in [3/4,5/4]\left(\sum_{f\in W} \frac{|b_e^T L_I^+ b_f|}{\sqrt{r_e}\sqrt{r_f}}\left(\sum_{k=1}^K Y_{ef}^{(k)}\right)\right)$$

for each $e\in W$. Since $Y_{ee}^{(k)} = 0$ for all $k$ and $e\in W$,

\begin{align*}
&\sum_{f\in W} \frac{|b_e^T L_I^+ b_f|}{\sqrt{r_e}\sqrt{r_f}}\left(\sum_{k=1}^K Y_{ef}^{(k)}\right)\\
&= \sum_{f\ne e\in W} \frac{|b_e^T L_I^+ b_f|}{\sqrt{r_e}\sqrt{r_f}}\left(\sum_{k=1}^K Y_{ef}^{(k)}\right)
\end{align*}

Notice that for $e\ne f$, $\{Y_{ef}^{(k)}\}_k$ is a family of independent Bernoullis with mean 1/4. Therefore, by Chernoff bounds and our choice of $K$, $K(1/4)(7/8)\le \sum_{k=1}^K Y_{ef}^{(k)} \le K(1/4)(9/8)$ for all $e\ne f$ with probability at least $1 - 1/n^5$. As a result,

$$\kappa_e\in \frac{K}{4}[1/2,3/2]\left(\sum_{f\ne e\in W} \frac{|b_e^T L_I^+ b_f|}{\sqrt{r_e}\sqrt{r_f}}\right)$$

with high probability, as desired.

\textbf{Runtime.} Lines \ref{line:sketch-mat} and \ref{line:recover-norm} contribute at most $\tilde{O}(|E(I)|)$ to the runtime of $\ColumnApx$ by the ``Runtime'' guarantee of Theorem \ref{thm:sketch}. Line \ref{line:compute-u} only takes $\tilde{O}(|E(I)|)$ time to compute $U$ because $C^T$ only has $O(\log n)$ columns. All other lines take linear time, so $\ColumnApx$ takes near-linear time.
\end{proof}

\subsubsection{Construction of concentrated edges}

Now, we subsample localized sets:

\begin{proposition}\label{prop:extra-local}
Given a graph $I$ and $\gamma\in (0,1)$, there is a set of edges $W\subseteq E(I)$ with two properties:

\begin{itemize}
\item (Size) $|W|\ge (\gamma / 4)|E(I)|$
\item (Value) For all $e\in W$, $\sum_{f\ne e\in W} \frac{|b_e^T L_I^+ b_f|}{\sqrt{r_e}\sqrt{r_f}}\le \psi$ for all $e\in W$, where $\psi := 100c_{\text{local}}\gamma (\log^2 n)$
\end{itemize}

Furthermore, there is an $\tilde{O}(|E(I)|/\gamma)$-expected time algorithm $\Subsample(I,\gamma)$ that produces $W$.
\end{proposition}

\begin{algorithm}[H]
\DontPrintSemicolon
\SetAlgoLined
\caption{$\Subsample(I,\gamma)$}

    \While{$W$ does not satisfy Proposition \ref{prop:extra-local}}{

        $W_0\gets $ random subset of $E(I)$, with each edge of $e\in E(I)$ included i.i.d. with probability $2\gamma$\;

        $(a_e)_{e\in W_0}\gets \ColumnApx(I,W_0)$\;

        $W\gets $ set of edges $e\in W_0$ with $a_e\le \psi/2$\;

    }

    \Return{$W$}

\end{algorithm}

\begin{proof}

We show that each iteration of the while loop terminates with probability at least $1/\text{polylog}(n)$. As a result, only $\text{polylog}(n)$ iterations are required to find the desired set. We do this by setting up an intermediate family of subsets of $E(I)$ to obtain $W$.

\textbf{Size.} Let $X_1\subseteq E(I)$ be the set of edges $e$ with $\sum_{f\in E(I)} \frac{|b_e^T L_I^+ b_f|}{\sqrt{r_e}\sqrt{r_f}}\le 2c_{\text{local}}\log^2 n$. By Theorem \ref{thm:localization}, $|X_1|\ge |E(I)|/2$.

Let $W_1 := X_1\cap W_0$. $W_1$ can alternatively be sampled by sampling $W_1$ from $X_1$, including each element of $X_1$ in $W_1$ i.i.d. with probability $2\gamma$. Furthermore,

\begin{align*}
\textbf{E}_{W_1}\left[\sum_{f\ne e\in W_0} \frac{|b_e^T L_I^+ b_f|}{\sqrt{r_e}\sqrt{r_f}} \Big| e\in W_1\right] &= \textbf{E}_{W_1}\left[\sum_{f\ne e\in W_0}\frac{|b_e^T L_I^+ b_f|}{\sqrt{r_e}\sqrt{r_f}}\right]\\
&= 2\gamma \sum_{f\ne e\in E(I)} \frac{|b_e^T L_I^+ b_f|}{\sqrt{r_e}\sqrt{r_f}}\\
&\le 4\gamma c_{\text{local}} (\log^2 n)\\
\end{align*}

\noindent By the approximation upper bound for $a_e$ and Markov's Inequality,

\begin{align*}
\Pr_{W_1}[e\notin W | e\in W_1] &\le \Pr_{W_1}[a_e > \psi/2 | e\in W_1]\\
&\le \Pr_{W_1}\left[\sum_{f\ne e\in E(I)} \frac{|b_e^T L_I^+ b_f|}{\sqrt{r_e}\sqrt{r_f}} > \psi/4 | e\in W_1\right]\\
&\le \frac{4\gamma c_{\text{local}} (\log^2 n)}{\psi/4}\\
&\le 1/2\\
\end{align*}

\noindent for every $e\in X_1$. Therefore, $$\textbf{E}[|W|] > (1/2)\textbf{E}[|W_1|] = \gamma |X_1| \ge \gamma |E(I)|/2$$ Since $0\le |W|\le |E(I)|$, $|W|\ge \gamma |E(I)|/4$ with probability at least $\gamma/4$, as desired.

\textbf{Value.} By the upper bound on $a_e$ due to Proposition \ref{prop:estimate-offdiag}, all edges $e\in W$ have the property that $$\sum_{f\ne e\in W} \frac{|b_e^T L_I^+ b_f|}{\sqrt{r_e}\sqrt{r_f}} \le \sum_{f\ne e\in W_0} \frac{|b_e^T L_I^+ b_f|}{\sqrt{r_e}\sqrt{r_f}}\le \psi$$ as desired.

\textbf{Runtime.} Each iteration of the while loop succeeds with probability at least $\gamma/4$, as discussed in the ``Size'' analysis. Each iteration takes $\tilde{O}(|E(I)|)$ time by the runtime guarantee for $\ColumnApx$. Therefore, the expected overall runtime is $\tilde{O}(|E(I)|/\gamma)$.

\end{proof}

\subsection{$\FastOracle$}

We now implement the $(\Theta(\log^3 n),\Theta(z/\log^3 n))$-steady oracle $\FastOracle$. It starts by finding a set $W$ guaranteed by Proposition \ref{prop:extra-local} with $\gamma = \Theta(1/\log^3 n)$. It then further restricts $W$ down to the set of edges satisfying ``Martingale change stability'' for $I_0$ and returns that set. The ``Value'' guarantee of Proposition \ref{prop:extra-local} ensures that these edges continue to satisfy the ``Martingale change stability'' guarantee even after conditioning on edges in $Z$.

\begin{algorithm}[H]
\SetAlgoLined
\DontPrintSemicolon
\caption{$\FastOracle(I,\mc S)$}

    \KwIn{a graph $I$ with leverage scores in $[3/16,13/16]$ and a subspace $\mc S\subseteq V(I)$ with $\mc S := \mathbb{R}^S\times \textbf{0}^{V(I)\setminus S}$ for some $S\subseteq V(I)$}

    \KwOut{a set $Z\subseteq E(I)$ satisfying the steady oracle guarantees}

    $W\gets \Subsample(I,\gamma)$, where $\gamma := 1/(100000000 c_{\text{local}} (\log^3 n)))$\;

    $\{\nu_e\}_{e\in E(I)}\gets \DiffApx(I,S,1/4,1/m^5)$\;

    \Return{all $e\in W$ for which $\nu_e\le \frac{4|S|}{|W|}$}

\end{algorithm}

To ensure that $\DiffApx$ is applicable, note the following equivalence to what its approximating and the quantity in the ``Martingale change stability'' guarantee:

\begin{proposition}\label{prop:fast-equiv}
$$\max_{x\in \mathbb{R}^S} \frac{(x^T L_H^+ b_f)^2}{r_f(x^T L_H^+ x)} = \frac{b_f^T L_H^+ b_f}{r_f} - \frac{b_f^T L_{H/S}^+ b_f}{r_f}$$
\end{proposition}

\begin{proof}
Define an $n \times (|S|-1)$ matrix $C$ with signed indicator vectors of edges in a star graph on $C$. For every $x\in \mathbb{R}^S$ with $x^T \textbf{1} = 0$, $x = C c_x$ for some unique $c_x\in \mathbb{R}^{|S|-1}$. Therefore,

\begin{align*}
\max_{x\in \mathbb{R}^S} \frac{(x^T L_H^+ b_f)^2}{r_f(x^T L_H^+ x)} &= \max_{c\in \mathbb{R}^{|S|-1}} \frac{c^T C^T L_H^+ b_f b_f^T L_H^+ C c}{r_f(c^T C^T L_H^+ C c)}\\
&= \frac{1}{r_f}\lambda_{\max}((C^T L_H^+ C)^{-1/2} C^T L_H^+ b_f b_f^T L_H^+ C (C^T L_H^+ C)^{-1/2})\\
&= \frac{1}{r_f}(b_f^T L_H^+ C (C^T L_H^+ C)^{-1} C^T L_H^+ b_f)\\
&= \frac{b_f^T L_H^+ b_f}{r_f} - \frac{b_f^T L_{H/S}^+ b_f}{r_f}\\
\end{align*}

where the last equality follows from the Woodbury formula.
\end{proof}

To analyze $\FastOracle$, we start by showing that any set of localized edges remains localized under random edge modifications:

\begin{proposition}\label{prop:maintain-local}
Consider a graph $I$ and a set of edges $Z\subseteq E(I)$ that satisfy the following two initial conditions:

\begin{itemize}
\item (Initial leverage scores) $\texttt{lev}_I(e)\in [3/16,13/16]$ for all $e\in Z$.
\item (Initial localization) $\sum_{f\in Z} \frac{|b_e^T L_I^+ b_f|}{\sqrt{r_e}\sqrt{r_f}}\le \tau$ for all $e\in Z$, where $\tau = \frac{1}{10000\log n}$.
\end{itemize}

Sample a sequence of minors $\{I_k\}_{k\ge 0}$ of $I$ and sets $Z_k\subseteq E(I_k)$ by letting $I_0 := I$ and for each $k \ge 0$, sampling a uniformly random edge $e_k\in Z_k$, letting $I_{k+1}\gets I_k\backslash e_k$ or $I_{k+1}\gets I_k/e_k$ arbitrarily, and letting $Z_{k+1}\gets Z_k\backslash e_k$. Then with probability at least $1 - 1/n^2$, the following occurs for all $i$:

\begin{itemize}
\item (All leverage scores) $\texttt{lev}_{I_k}(e)\in [1/8,7/8]$ for all $e\in Z_k$.
\item (All localization) $\sum_{f\in Z_k,f\ne e} \frac{|b_e^T L_{I_k}^+ b_f|}{\sqrt{r_e}\sqrt{r_f}} \le \tau'$ for all $e\in Z_k$, where $\tau' = 2\tau$.
\end{itemize}
\end{proposition}

To prove this result, we cite the following submartingale inequality:

\begin{theorem}[Theorem 27 of \cite{CL06} with $a_i = 0$ for all $i$]\label{thm:one-sided-mart}
Let $(Y_i)_{i\ge 0}$ be a submartingale with difference sequence $X_i := Y_i - \textbf{E}[Y_i | Y_{i-1}]$ and $W_i := \sum_{j=1}^i \textbf{E}[X_j^2|Y_{j-1}]$. Suppose that both of the following conditions hold for all $i\ge 0$:

\begin{itemize}
\item $W_i\le \sigma^2$
\item $X_i\le M$
\end{itemize}

Then $$\Pr[Y_i - Y_0 \ge \lambda]\le \exp\left(-\frac{\lambda^2}{2(\sigma^2 + M\lambda/3)}\right)$$
\end{theorem}

\begin{proof}[Proof of Proposition \ref{prop:maintain-local}]
We prove this result by induction on $k$. For $k = 0$, ``Initial leverage scores'' and ``Initial localization'' imply ``All leverage scores'' and ``All localization'' respectively. For $k > 0$, we use submartingale concentration to show the inductive step. For any edge $e\in Z_k$, define two random variables $U_e^{(k)} := \texttt{lev}_{I_k}(e)$ and $V_e^{(k)} := \sum_{f\in Z_k,f\ne e} \frac{|b_e^T L_{I_k}^+ b_f|}{\sqrt{r_e}\sqrt{r_f}}$. Let $$\widehat{U}_e^{(k)} := U_e^{(k)} - \sum_{l=0}^{k-1} \textbf{E}[U_e^{(l+1)} - U_e^{(l)}|U_e^{(l)}]$$ and $$\widehat{V}_e^{(k)} := V_e^{(k)} - \sum_{l=0}^{k-1} \textbf{E}\left[V_e^{(l+1)} - V_e^{(l)} + \frac{|b_e^T L_{I_l}^+ b_{e_l}|}{\sqrt{r_e}\sqrt{r_{e_l}}} \Big| V_e^{(l)}\right]$$ $(\widehat{U}_e^{(k)})_{k\ge 0}$ is a martingale and $(\widehat{V}_e^{(k)})_{k\ge 0}$ is a submartingale for all $e\in Z_k$. Let $$\widehat{XU}_e^{(k)} := \widehat{U}_e^{(k)} - \textbf{E}[\widehat{U}_e^{(k)}|\widehat{U}_e^{(k-1)}] = U_e^{(k)} - U_e^{(k-1)} - \textbf{E}[U_e^{(k)} - U_e^{(k-1)}|U_e^{(k-1)}]$$ $$\widehat{XV}_e^{(k)} := \widehat{V}_e^{(k)} - \textbf{E}[\widehat{V}_e^{(k)}|\widehat{V}_e^{(k-1)}] = V_e^{(k)} - V_e^{(k-1)} - \textbf{E}[V_e^{(k)} - V_e^{(k-1)} | V_e^{(k-1)}]$$ $$\widehat{WU}_e^{(k)} := \sum_{j=1}^k \textbf{E}[(\widehat{XU}_e^{(j)})^2|\widehat{U}_e^{(j-1)}]$$ and $$\widehat{WV}_e^{(k)} := \sum_{j=1}^k \textbf{E}[(\widehat{XV}_e^{(j)})^2|\widehat{V}_e^{(j-1)}]$$ By Sherman-Morrison and the inductive assumption applied to the edges $e,e_{k-1}\in Z_{k-1}$,

\begin{align*}
|\widehat{XU}_e^{(k)}| &\le |U_e^{(k)} - U_e^{(k-1)}| + \textbf{E}[|U_e^{(k)} - U_e^{(k-1)}| | U_e^{(k-1)}]\\
&\le 2\frac{(b_e^T L_{I_{k-1}}^+ b_{e_{k-1}})^2}{r_e \min(b_{e_{k-1}}^T L_{I_{k-1}}^+ b_{e_{k-1}}, r_{e_{k-1}} - b_{e_{k-1}}^T L_{I_{k-1}}^+ b_{e_{k-1}})}\\
&\le 16(\tau')^2\\
\end{align*}

\begin{align*}
\widehat{XV}_e^{(k)} &= \left(V_e^{(k)} - \left(V_e^{(k-1)} - \frac{|b_e^T L_{I_{k-1}}^+ b_{e_{k-1}}|}{\sqrt{r_e}\sqrt{r_{e_{k-1}}}}\right)\right) - \textbf{E}\left[V_e^{(k)} - \left(V_e^{(k-1)} - \frac{|b_e^T L_{I_{k-1}}^+ b_{e_{k-1}}|}{\sqrt{r_e}\sqrt{r_{e_{k-1}}}}\right)\right]\\
&\le 2\sum_{g\in Z_k,g\ne e} \left|\frac{|b_e^T L_{I_k}^+ b_g|}{\sqrt{r_e}\sqrt{r_g}} - \frac{|b_e^T L_{I_{k-1}}^+ b_g|}{\sqrt{r_e}\sqrt{r_g}}\right|\\
&\le 2\sum_{g\in Z_k,g\ne e} \frac{|b_e^T L_{I_{k-1}}^+ b_{e_{k-1}}||b_{e_{k-1}}^T L_{I_{k-1}}^+ b_g|}{\sqrt{r_e}r_{e_{k-1}}\min(1 - \texttt{lev}_{I_{k-1}}(e_{k-1}), \texttt{lev}_{I_{k-1}}(e_{k-1}))\sqrt{r_g}}\\
&\le 16 \frac{|b_e^T L_{I_{k-1}}^+ b_{e_{k-1}}|}{\sqrt{r_e}\sqrt{r_{e_{k-1}}}} \sum_{g\in Z_{k-1},g\ne e} \frac{|b_{e_{k-1}}^T L_{I_{k-1}}^+ b_g|}{\sqrt{r_g}\sqrt{r_{e_{k-1}}}}\\
&\le 16(\tau')^2\\
\end{align*}

\begin{align*}
&\textbf{E}[(\widehat{XU}_e^{(k)})^2|\widehat{U}_e^{(k-1)},e\ne e_{k-1},\hdots,e\ne e_0]\\
&\le 4\textbf{E}[(U_e^{(k)} - U_e^{(k-1)})^2 | U_e^{(k-1)},e\ne e_{k-1},\hdots,e\ne e_0]\\
&\le 4\textbf{E}_{e_{k-1}}\left[\frac{(b_e^T L_{I_{k-1}}^+ b_{e_{k-1}})^4}{r_e^2 r_{e_{k-1}}^2 \min(1 - \texttt{lev}_{I_{k-1}}(e_{k-1}), \texttt{lev}_{I_{k-1}}(e_{k-1}))^2} \Big| e\ne e_{k-1}\right]\\
&\le 256\textbf{E}_{e_{k-1}}\left[\frac{(b_e^T L_{I_{k-1}}^+ b_{e_{k-1}})^4}{r_e^2 r_{e_{k-1}}^2} \Big| e\ne e_{k-1}\right]\\
&\le \frac{256}{|Z_{k-1}| - 1}\left(\sum_{f\in Z_{k-1},f\ne e}\frac{|b_e^T L_{I_{k-1}}^+ b_f|}{\sqrt{r_e} \sqrt{r_f}}\right)^4\\
&\le \frac{256(\tau')^4}{|Z_{k-1}| - 1}\\
\end{align*}

\begin{align*}
&\textbf{E}[(\widehat{XV}_e^{(k)})^2|\widehat{V}_e^{(k-1)}, e\ne e_{k-1},\hdots,e\ne e_0]\\
&\le 4\textbf{E}\left[\left(V_e^{(k)} - \left(V_e^{(k-1)} - \frac{|b_e^T L_{I_{k-1}}^+ b_{e_{k-1}}|}{\sqrt{r_e}\sqrt{r_{e_{k-1}}}}\right)\right)^2 \Big| V_e^{(k-1)}, e\ne e_{k-1},\hdots,e\ne e_0\right]\\
&\le 4\textbf{E}_{e_{k-1}}\left[\left(\sum_{g\in Z_k,g\ne e} \frac{|b_e^T L_{I_{k-1}}^+ b_{e_{k-1}}||b_{e_{k-1}}^T L_{I_{k-1}}^+ b_g|}{\sqrt{r_e}r_{e_{k-1}}\min(1 - \texttt{lev}_{I_{k-1}}(e_{k-1}), \texttt{lev}_{I_{k-1}}(e_{k-1}))\sqrt{r_g}}\right)^2 \Big| e\ne e_{k-1}\right]\\
&\le 256\textbf{E}_{e_{k-1}}\left[\left(\frac{(b_e^T L_{I_{k-1}}^+ b_{e_{k-1}})^2}{r_er_{e_{k-1}}}\right)\left(\sum_{g\in Z_k,g\ne e} \frac{|b_{e_{k-1}}^T L_{I_{k-1}}^+ b_g|}{\sqrt{r_{e_{k-1}}}\sqrt{r_g}}\right)^2 \Big| e\ne e_{k-1}\right]\\
&\le \frac{256}{|Z_{k-1}|-1} \max_{f\in Z_{k-1}} \left(\sum_{g\in Z_{k-1},g\ne f} \frac{|b_f^T L_{I_{k-1}}^+ b_g|}{\sqrt{r_f}\sqrt{r_g}}\right)^4\\
&\le \frac{256(\tau')^4}{|Z_{k-1}| - 1}\\
\end{align*}

\noindent Therefore, for all $k\le |Z|/2$, $|\widehat{WU}_e^{(k)}|\le 256(\tau')^4$ and $|\widehat{WV}_e^{(k)}|\le 256(\tau')^4$ given the inductive hypothesis. By Theorem \ref{thm:one-sided-mart}, $$\Pr[|\widehat{U}_k - \widehat{U}_0| > 2000(\log n)(\tau')^2] \le \exp\left(-\frac{(2000(\log n)(\tau')^2)^2}{512(\tau')^4 + 512(\tau')^2(2000(\log n)(\tau')^2/3)}\right)\le \frac{1}{n^5}$$ and $$\Pr[\widehat{V}_k - \widehat{V}_0 > 2000(\log n)(\tau')^2] \le \exp\left(-\frac{(2000(\log n)(\tau')^2)^2}{512(\tau')^4 + 512(\tau')^2(2000(\log n)(\tau')^2/3)}\right)\le \frac{1}{n^5}$$ Now, we bound $U_k - \widehat{U}_k$ and $V_k - \widehat{V}_k$. By Sherman-Morrison and the inductive assumption for $Z_{k-1}$,

\begin{align*}
\textbf{E}[|U_e^{(k)} - U_e^{(k-1)}| | U_e^{(k-1)},e\ne e_{k-1}] &\le \textbf{E}_{e_{k-1}}\left[\frac{(b_e^T L_{I_{k-1}}^+ b_{e_{k-1}})^2}{r_e\min(1 - \texttt{lev}_{I_{k-1}}(e_{k-1}), \texttt{lev}_{I_{k-1}}(e_{k-1})) r_{e_{k-1}}} \Big| e\ne e_{k-1}\right]\\
&\le \frac{8(\tau')^2}{|Z_{k-1}|-1}\\
\end{align*}

\noindent and

\begin{align*}
&\textbf{E}\left[V_e^{(k)} - V_e^{(k-1)} + \frac{|b_e^T L_{I_{k-1}}^+ b_{e_{k-1}}|}{\sqrt{r_e}\sqrt{r_{e_{k-1}}}} \Big| V_e^{(k-1)},e\ne e_{k-1}\right]\\
&\le \textbf{E}_{e_{k-1}}\left[\sum_{g\in Z_k,g\ne e} \frac{|b_e^T L_{I_{k-1}}^+ b_{e_{k-1}}||b_{e_{k-1}}^T L_{I_{k-1}}^+ b_g|}{\sqrt{r_e}r_{e_{k-1}}\min(1 - \texttt{lev}_{I_{k-1}}(e_{k-1}), \texttt{lev}_{I_{k-1}}(e_{k-1}))\sqrt{r_g}} \Big| e\ne e_{k-1}\right]\\
&\le \frac{8(\tau')^2}{|Z_{k-1}|-1}\\
\end{align*}

\noindent so for $k\le |Z|/2$, $|U_k - \widehat{U}_k|\le 8(\tau')^2$ and $V_k - \widehat{V}_k\le 8(\tau')^2$. In particular, with probability at least $1 - 2/n^5$,

\begin{align*}
|\texttt{lev}_{I_k}(e) - \texttt{lev}_{I_0}(e)| &= |U_e^{(k)} - U_e^{(0)}|\\
&\le |U_e^{(k)} - \widehat{U}_e^{(k)}| + |\widehat{U}_e^{(k)} - \widehat{U}_e^{(0)}| + |\widehat{U}_e^{(0)} - U_e^{(0)}|\\
&\le 8(\tau')^2 + 2000(\log n)(\tau')^2 + 0\\
&\le 1/16\\
\end{align*}

\noindent Therefore, $\texttt{lev}_{I_k}(e)\in [1/8,7/8]$ with probability at least $1 - 2/n^5$ for all $e\in Z_k$. Furthermore,

\begin{align*}
\sum_{g\in Z_k, g\ne e} \frac{|b_e^T L_{I_k}^+ b_g|}{\sqrt{r_e}\sqrt{r_g}}&= V_e^{(k)}\\
&= (V_e^{(k)} - \widehat{V}_e^{(k)}) + (\widehat{V}_e^{(k)} - \widehat{V}_e^{(0)}) + (\widehat{V}_e^{(0)} - V_e^{(0)}) + V_e^{(0)}\\
&\le 8(\tau')^2 + 2000(\log n)(\tau')^2 + 0 + \tau\\
&\le 2\tau = \tau'\\
\end{align*}

\noindent This completes the inductive step and the proof of the proposition.

\end{proof}

Now, we prove Theorem \ref{thm:schur-fast}. By Lemma \ref{lem:main-oracle}, it suffices to show that $\FastOracle$ is a $(O(\log^3 n),\Omega(z/\log^3 n)))$-steady oracle with runtime $\tilde{O}(|E(I)|)$.

\begin{proof}[Proof of Theorem \ref{thm:schur-fast}]

\textbf{Size of $Z$.} By Proposition \ref{prop:sum-bound}, Proposition \ref{prop:fast-equiv}, and the approximation upper bound for $\nu_e$, $$\sum_{e\in E(I)} \nu_e\le (5/4)|S| + 1/m^4\le (3/2)|S|$$ Therefore, by Markov's Inequality, $|Z|\ge 5|W|/8$. By the ``Size'' guarantee of Proposition \ref{prop:extra-local}, $|W|\ge \Omega(1/(\log n)^3) |E(I)|$, so $|Z|\ge \Omega(1/(\log n)^3) |E(I)|$, as desired.

\textbf{Leverage score stability.} We start by checking that the input conditions for Proposition \ref{prop:maintain-local} are satisfied for $Z$. The ``Initial leverage scores'' condition is satisfied thanks to the ``Leverage scores'' input guarantee for steady oracles. The ``Initial localization'' condition is satisfied because of the ``Value'' output guarantee for $\Subsample$ applied to $W$. Therefore, Proposition \ref{prop:maintain-local} applies. The ``All leverage scores'' guarantee of Proposition \ref{prop:maintain-local} is precisely the ``Leverage score stability'' guarantee of steady oracles, as desired.

\textbf{Martingale change stability.} Let $(y_i)_{i=1}^d$ be a basis of $\mc S$ for which $y_i L_I^+ y_j = 0$ for $i\ne j$ and $y_i L_I^+ y_i$ for all $i\in [\Dim(\mc S)]$. Let $Y_t$ be a $V(I_t)\times \Dim(\mc S)$ matrix with columns $(y_i)_{I_t}$. By Proposition \ref{prop:equiv} applied with $G\gets I$, $$\max_{x\in \mc S} \frac{(x_{I_t}^T L_{I_t}^+ b_f)^2}{r_f(x^T L_I^+ x)} = \frac{b_f^T L_{I_t}^+ Y_t Y_t^T L_{I_t}^+ b_f}{r_f}$$ for all $f\in E(I_t)$. We now bound this quantity over the course of deletions and contractions of the edges $e_i$ by setting up a martingale. For all $t\ge 0$ and $f\in E(I)$, let $$A_f^{(t)} := \frac{b_f^T L_{I_t}^+ Y_t Y_t^T L_{I_t}^+ b_f}{r_f}$$ and $$\widehat{A}_f^{(t)} := A_f^{(t)} - \sum_{s=0}^{t-1}[A_f^{(s+1)} - A_f^{(s)} | A_f^{(s)}]$$ For each $f\in E(I)$, $(A_f^{(t)})_{t\ge 0}$ is a martingale. Let $$\widehat{XA}_f^{(t)} := \widehat{A}_f^{(t)} - \widehat{A}_f^{(t-1)} = A_f^{(t)} - A_f^{(t-1)} - \textbf{E}[A_f^{(t)} - A_f^{(t-1)} | A_f^{(t-1)}]$$ and $$\widehat{WA}_f^{(t)} := \sum_{s=1}^t \textbf{E}[(\widehat{XA}_f^{(s)})^2 | A_f^{(s-1)}]$$ We now inductively show that for all $f\in Z_t$ (which includes $f_t$), $$A_f^{(t)}\le \frac{\xi' |S|}{|E(I)|}$$ where $\xi := 8 \frac{E(I)}{|W|} \le O(\log^3 n)$ and $\xi' := 2\xi$. Initially, $$A_f^{(0)}\le \frac{8|S|}{|W|} = \frac{\xi |S|}{|W|}$$ for all $f\in Z$ by the approximation lower bound for $\nu_e$, completing the base case. For $t > 0$, we bound $A_f^{(t)}$ for $f\in Z_t$ by using martingale concentration. We start by bounding differences using the ``All leverage scores'' guarantee of Proposition \ref{prop:maintain-local}, Sherman-Morrison, Cauchy-Schwarz, and the inductive assumption:

\begin{align*}
|A_f^{(t)} - A_f^{(t-1)}| &= \left|\frac{b_f^T L_{I_t}^+ Y_t Y_t^T L_{I_t}^+ b_f}{r_f} - \frac{b_f^T L_{I_{t-1}}^+ Y_{t-1} Y_{t-1}^T L_{I_{t-1}}^+ b_f}{r_f}\right|\\
&\le 2\left|\frac{b_f^T L_{I_{t-1}}^+ b_{f_{t-1}} b_{f_{t-1}}^T L_{I_{t-1}}^+ Y_{t-1} Y_{t-1}^T L_{I_{t-1}}^+ b_f}{r_{f_{t-1}}\min(1 - \texttt{lev}_{I_{t-1}}(f_{t-1}), \texttt{lev}_{I_{t-1}}(f_{t-1}))r_f}\right|\\
&+ \left|\frac{b_f^T L_{I_{t-1}}^+ b_{f_{t-1}} b_{f_{t-1}}^T L_{I_{t-1}}^+ Y_{t-1} Y_{t-1}^T L_{I_{t-1}}^+ b_{f_{t-1}} b_{f_{t-1}}^T L_{I_{t-1}}^+ b_f}{r_{f_{t-1}}^2(\min(1 - \texttt{lev}_{I_{t-1}}(f_{t-1}), \texttt{lev}_{I_{t-1}}(f_{t-1})))^2r_f}\right|\\
&\le 16\frac{|b_f^T L_{I_{t-1}}^+ b_{f_{t-1}}| |b_{f_{t-1}}^T L_{I_{t-1}}^+ Y_{t-1} Y_{t-1}^T L_{I_{t-1}}^+ b_f|}{r_{f_{t-1}}r_f}\\
&+ 64\frac{(b_f^T L_{I_{t-1}}^+ b_{f_{t-1}})^2 b_{f_{t-1}}^T L_{I_{t-1}}^+ Y_{t-1} Y_{t-1}^T L_{I_{t-1}}^+ b_{f_{t-1}}}{r_{f_{t-1}}^2r_f}\\
&\le 16\frac{|b_f^T L_{I_{t-1}}^+ b_{f_{t-1}}|}{\sqrt{r_{f_{t-1}}}\sqrt{r_f}} \sqrt{A_{f_{t-1}}^{(t-1)}}\sqrt{A_f^{(t-1)}}\\
&+ 64\left(\frac{|b_f^T L_{I_{t-1}}^+ b_{f_{t-1}}|}{\sqrt{r_f}\sqrt{r_{f_{t-1}}}}\right)^2 A_{f_{t-1}}^{(t-1)}\\
&\le 80\frac{|b_f^T L_{I_{t-1}}^+ b_{f_{t-1}}|}{\sqrt{r_{f_{t-1}}}\sqrt{r_f}}\frac{\xi' |S|}{|E(I)|}\\
\end{align*}

\noindent By the ``All localization'' guarantee of Proposition \ref{prop:maintain-local},

\begin{align*}
|\widehat{XA}_f^{(t)}|&\le |A_f^{(t)} - A_f^{(t-1)}| + \textbf{E}[|A_f^{(t)} - A_f^{(t-1)}| | A_f^{t-1}]\\
&\le 160\tau' \frac{\xi' |S|}{|E(I)|}\\
\end{align*}

\noindent and

\begin{align*}
\textbf{E}[(\widehat{XA}_f^{(t)})^2 | \widehat{A}_f^{(t-1)}] &\le 4\textbf{E}_{f_{t-1}}[(A_f^{(t)} - A_f^{(t-1)})^2 | A_f^{(t-1)}, f\ne f_{t-1}]\\
& \frac{6400(\tau')^2}{|Z_{t-1}|-1} \left(\frac{\xi' |S|}{|E(I)|}\right)^2\\
\end{align*}

Since $K(|Z|)\le |Z|/2$, $|\widehat{WA}_f^{(t)}|\le 6400(\tau')^2 \left(\frac{\xi' |S|}{|E(I)|}\right)^2$. Therefore, by Theorem \ref{thm:one-sided-mart} applied to the submartingales $(\widehat{A}_f^{(t)})_{t\ge 0}$ and $(-\widehat{A}_f^{(t)})_{t\ge 0}$, $$\Pr\left[|\widehat{A}_f^{(t)} - \widehat{A}_f^{(0)}| > \frac{\xi' |S|}{5|E(I)|}\right] \le \exp\left(-\frac{((\xi' |S|)/(5|E(I)|))^2}{(6400(\tau')^2 + 160(\tau'))((\xi' |S|)/(|E(I)|))^2}\right) \le 1/n^5$$ Since $\widehat{A}_f^{(0)} = A_f^{(0)}$, we just need to bound $|A_f^{(t)} - \widehat{A}_f^{(t)}|$. We do this by bounding expectations of differences:

\begin{align*}
\textbf{E}[|A_f^{(t)} - A_f^{(t-1)}| | A_f^{(t-1)}, f\ne f_{t-1}] &\le 80\textbf{E}\left[\frac{|b_f^T L_{I_{t-1}}^+ b_{f_{t-1}}|}{\sqrt{r_{f_{t-1}}}\sqrt{r_f}}\frac{\xi' |S|}{|E(I)|} \Big| f\ne f_{t-1}\right]\\
&\le \frac{80\tau'}{|Z_{t-1}|-1} \frac{\xi' |S|}{|E(I)|}\\
\end{align*}

Therefore, $|A_f^{(t)} - \widehat{A}_f^{(t)}|\le \sum_{s=0}^{t-1}\textbf{E}[|A_f^{(s+1)} - A_f^{(s)}| | A_f^{(s)}, f\ne f_s] \le \frac{\xi' |S|}{5|E(I)|}$. This means that $$A_f^{(t)}\le |A_f^{(t)} - \widehat{A}_f^{(t)}| + |\widehat{A}_f^{(t)} - \widehat{A}_f^{(0)}| + A_f^{(0)}\le \frac{2\xi' |S|}{5|E(I)|} + \frac{\xi |S|}{|E(I)|}\le \frac{\xi' |S|}{|E(I)|}$$ with probability at least $1 - 1/n^5$, which completes the inductive step.

Therefore, by a union bound and the fact that $f_t\in Z_t$ for all $t\ge 0$, $$\max_{x\in \mc S} \frac{(x_{I_t} L_{I_t}^+ b_{f_t})^2}{r_{f_t} (x^T L_I^+ x)} = A_{f_t}^{(t)}\le \frac{\xi' |S|}{|E(I)|} \le \frac{O(\log^3 n)|S|}{|E(I)|}$$ completing the ``Martingale change stability'' proof.

\end{proof}

\section{Efficient approximation of differences}

%\todo{expand this section}

%To prove Lemma \ref{lem:diff-apx}, do not identify all vertices in $S$ all at once. Instead, add a very high (polynomially high) resistance star graph between the vertices in $S$. This does not change effective resistances by more than a very small factor. Then, iteratively half-contract edges in the star until an edge has leverage score greater than 1/2. At this point, completely contract it.

In this section, we show how to approximate changes
in effective resistances due to the identification of a given vertex set $S$,
and thus prove Lemma~\ref{lem:diff-apx}.
Namely, given a vertex set $S \subset V$,
we need to approximate the following quantity for all edges $e\in E(G)$:
\[
    (b_e^T L_G^+ b_e) - (b_e^T L_{G/S}^+ b_e).
\]
By a proof similar to that of Proposition~\ref{prop:fast-equiv}, this quantity equals
\begin{align}\label{eq:max}
    \max_{x \bot 1,x\in \mathbb{R}^S} \frac{(x^T L_G^+ b_e)^2}{x^T L_G^+ x},
\end{align}
where $1$ denotes the all-one vector.
\begin{lemma}
    The decrease in the effective resistance of an edge $e\in E(G)$
    due to the identification of a vertex set $S\subset V$ equals
    \begin{align*}
        (b_e^T L_G^+ b_e) - (b_e^T L_{G/S}^+ b_e) =
        \max_{x \bot 1,x\in \mathbb{R}^S} \frac{(x^T L_G^+ b_e)^2}{x^T L_G^+ x}.
    \end{align*}
\end{lemma}
\begin{proof}
    Let $C$ be the $n \times (|S|-1)$ matrix with signed indicator vectors of edges
    in a star graph supported on $S$.
    Then we have
    \begin{align*}
        & (b_e^T L_G^+ b_e) - (b_e^T L_{G/S}^+ b_e) \\
        = & b_e^T L_H^+ C (C^T L_H^+ C)^{-1} C^T L_H^+ b_e
        \qquad \text{by Woodbury} \\
        = & \lambda_{\max}((C^T L_H^+ C)^{-1/2} C^T L_H^+ b_e b_e^T L_H^+ C (C^T L_H^+ C)^{-1/2})\\
        = & \max_{c\in \mathbb{R}^{|S|-1}} \frac{c^T C^T L_H^+ b_e b_e^T L_H^+ C c}{c^T C^T L_H^+ C c}\\
        = & \max_{x\bot 1,x\in \mathbb{R}^S} \frac{(x^T L_G^+ b_e)^2}{x^T L_H^+ x},
    \end{align*}
    where the last equality follows
    from that the columns of $C$ form a basis of the subspace
    of $\mathbb{R}^S$ orthogonal to the all-ones vector.
\end{proof}

Let $k := \sizeof{S}$,
and suppose without loss of generality that
$S$ contains the first $k$ vertices in $G$.
We construct a basis (plus an extra vector) of the subspace of $\mathbb{R}^S$
orthogonal to the all-ones vector by letting
\begin{align}
    C_{n\times k} := \begin{pmatrix}
        I_{k \times k} - \frac{1}{k} J_{k\times k} \\
        0_{(n-k)\times k}
    \end{pmatrix},
    \label{eq:basis}
\end{align}
where $I$ denotes the identity matrix,
and $J$ denotes the matrix whose entries are all $1$.
Let $P_{n\times k} := \begin{pmatrix} I_{k\times k} & 0 \end{pmatrix}^T$
be the projection matrix taking the first $k$ coordinates,
and let $\Pi_{k\times k} := I_{k\times k} - \frac{1}{k} J_{k\times k}$.
Now we can write~(\ref{eq:max}) as
\begin{align}
    &\max_{x \bot 1,x\in \mathbb{R}^S} \frac{(x^T L_G^+ b_e)^2}{x^T L_G^+ x} \notag \\
    = &\max_{c\in \mathbb{R}^{k}} \frac{c^T C^T L_G^+ b_e b_e^T L_G^+ C c}{c^T C^T L_G^+ C c} \notag \\
    = &\max_{c\in \mathbb{R}^{k}}
    \frac{\left(c^T \Pi_{k\times k}\right)\ C^T L_G^+ b_e b_e^T L_G^+ C\ \left(\Pi_{k\times k} c\right)}
    {\left(c^T \Pi_{k\times k}\right)\ C^T L_G^+ C\ \left(\Pi_{k\times k} c\right)}
    \qquad \text{by $C \Pi_{k\times k} = C$} \notag\\
    = &\max_{c\in \mathbb{R}^{k}}
    \frac{\left(c^T (C^T L_G^+ C)^{+/2}\right)\ C^T L_G^+ b_e b_e^T L_G^+ C\ \left((C^T L_G^+ C)^{+/2}c\right)}
    {\left(c^T (C^T L_G^+ C)^{+/2}\right)\  C^T L_G^+ C \ \left((C^T L_G^+ C)^{+/2}c\right)}
    \notag\\
    &\qquad \text{since $(C^T L_G^+ C)^{+/2}$ and $\Pi_{k\times k}$ have the same column space} \notag\\
    = &\lambda_{\max}((C^T L_G^+ C)^{+/2} C^T L_G^+ b_e b_e^T L_G^+ C (C^T L_G^+ C)^{+/2})\notag\\
    = &b_e^T L_G^+ C (C^T L_G^+ C)^{+} C^T L_G^+ b_e\notag\\
    = &b_e^T L_G^+ C (\Pi_{k\times k} P^T L_G^+ P \Pi_{k\times k})^{+} C^T L_G^+ b_e
    \qquad \text{by $P \Pi_{k\times k} = C$} \notag \\
    = &b_e^T L_G^+ C\ SC(L_G, S)\ C^T L_G^+ b_e
    \qquad
    \text{by Fact~\ref{fact:schurinv}.}\label{eq:difsc}
    %\text{by $(SC(L_G,S))^+ = \Pi_{k\times k} (L_G^+)_{S,S} \Pi_{k\times k}$} \label{eq:difsc}
\end{align}
%where $SC(L_G,S)$ denotes the Schur complement of $L_G$ onto $S$.
%Here
%the second equality follows from the identity $C \Pi_{k\times k} = C$.
%The third equality follows from the fact
%that $\Pi_{k\times k}$ and $C^T L_G^+ C$
%have the same row and column spaces.
%that the subspace of $\mathbb{R}^k$ orthogonal to the all-ones vector
%is identical to the row/column space of $C^T L_G^+ C$.
%The sixth equality follows from the identity
%$P\Pi_{k\times k} = C$.
%The last equality follows from the fact
%that $(SC(L_G,S))^+ = \Pi_{k\times k} (L_G^+)_{S,S} \Pi_{k\times k}$.

%\hl{We can avoid expanding the following by computing the Schur complement $SC(L_G,S)$
%using approximate elimination as in~\cite{KS16}, which might make the error tracking easier.}

To approximate~(\ref{eq:difsc}), we further write it as
\begin{align*}
    &b_e^T L_G^+ C\ SC(L_G, S)\ C^T L_G^+ b_e \\
    = &b_e^T L_G^+ C\ SC(L_G, S) (SC(L_G, S))^+ SC(L_G, S)\ C^T L_G^+ b_e \\
    = &b_e^T L_G^+ C\ SC(L_G, S)\ C^T L_G^+ C\ SC(L_G, S)\ C^T L_G^+ b_e \\
    = &b_e^T L_G^+ C\ SC(L_G, S)\ C^T L_G^+ (B_G^T W_G B_G) L_G^+ C\ SC(L_G, S)\ C^T L_G^+ b_e,
\end{align*}
where the last equality follows from $L_G^+ = L_G^+ L_G L_G^+$ and $L_G = B_G^T W_G B_G$.

We now write the change in the effective resistance of an edge $e$
in a square of an Euclidean norm as
\begin{align*}
    b_e^T L_G^+ b_e - b_e^T L_{G/S}^+ b_e
    =
    \norm{W_G^{1/2} B_G L_G^+ C (SC(L_G, S)) C^T L_G^+ b_e}^2.
\end{align*}
We then use Johnson-Lindenstrauss Lemma to reduce dimensions.
Let $Q_{k\times m}$ be a random $\pm 1$ matrix where $k \geq 24 \log n / \ep^2$.
By Lemma~\ref{lem:jl}, the following statement holds for all $e$ with high probability:
\begin{align}
    \norm{W_G^{1/2} B_G L_G^+ C (SC(L_G, S)) C^T L_G^+ b_e}^2
    \approx_{1 + \ep}
    \norm{Q W_G^{1/2} B_G L_G^+ C (SC(L_G, S)) C^T L_G^+ b_e}^2.
    \label{eq:emb}
\end{align}

To compute the matrix on the rhs,
we note that $C$ is easy to apply by applying $I$ and $J$,
and $L_G^+$ can be applied to high accuracy by Fast Laplacian Solvers. %~\cite{ST14, CKMPPRX14}.
Thus, we only need to apply the Schur complement $SC(L_G,S)$ to high accuracy fast.
We recall Definition~\ref{def:schur} of Schur complements:
\begin{align*}
    SC(L_G,S) := (L_G)_{S,S} - (L_G)_{S,T} (L_G)_{T,T}^{-1} (L_G)_{T,S},
\end{align*}
where $T := V\setminus S$.
Since $(L_G)_{T,T}$ is a principle submatrix of $L_G$,
it is an SDDM matrix and hence its inverse can be applied also by Fast Laplacian Solvers
to high accuracy.
%We characterize the performance of Fast Laplacian Solvers in the following lemma.

\subsection{The subroutine and proof of Lemma~\ref{lem:diff-apx}}

We give the algorithm for approximating changes in effective resistances due to the identification
of $S$ as follows:

\begin{algorithm}[H]
    \SetAlgoLined
    \DontPrintSemicolon

    \KwIn{A weighted graph $G$, a set of vertices $S \subseteq V(G)$, and $\delta_0, \delta_1\in (0,1)$}

    \KwOut{Estimates $\left\{\nu_e\right\}_{e\in E(G)}$ to
    differences in effective resistances in $G$ and $G/S$}

    Let $Q_{k\times m}$ be a random $\pm 1$ matrix
    where $k \geq 24 \log n / \delta_0^{2}$.

    Compute each row of $Y_{k\times n} := Q W_G^{1/2} B_G L_G^+ C (SC(L_G, S)) C^T L_G^+$
    by applying $L_G^+$ and $L_{V\setminus S,V\setminus S}^{-1}$ to accuracy
    \[
        \eps = \frac{\delta_1}
        {48 \sqrt{k} \cdot n^{8.5}\cdot \wmax^{2.5} \wmin^{-3}}.
    \]\;
    $\nu_e \gets \norm{Y b_e}^2$ for all $e \in E(G)$ \;
    \Return{$\left\{\nu_e\right\}_{e\in E(G)}$}
    \caption{$\DiffApx(G, S, \delta_0, \delta_1)$}
\end{algorithm}

To prove the approximation ratio for $\DiffApx$, we first
track the errors for applying Schur complement in
the following lemma:

\begin{lemma}\label{lem:applyschur}
    For any Laplacian $L_G$, $S\subset V(G)$,
    vector $b\in \mathbb{R}^n$, and $\eps > 0$,
    the following statement holds:
    \[
        \norm{x - \xtil} \leq \eps n^{3.5} \wmax^{2.5} \wmin^{-0.5} \norm{b},
    \]
    where
    \begin{align*}
        &x := \kh{(L_G)_{S,S} - (L_G)_{S,T} (L_G)_{T,T}^{-1} (L_G)_{T,S}} b, \\
        &\xtil := (L_G)_{S,S} b - (L_G)_{S,T} \xtil_1, \\
        &\xtil_1 := \LaplSolve((L_G)_{T,T}, (L_G)_{T,S}b, \ep).
    \end{align*}
\end{lemma}

Using Lemma~\ref{lem:applyschur},
we track the errors for computing the embedding in~(\ref{eq:emb})
as follows:

\begin{lemma}\label{lem:computeemb}
    For any Laplacian $L_G$, $S\subset V(G)$,
    vector $q\in \mathbb{R}^n$ with entries $\pm 1$, and
    \begin{align}\label{eq:eps}
        0 < \eps < 1 / \kh{4n^{6} \cdot \wmax^{2.5} \wmin^{-1.5}},
    \end{align}
    and a matrix $C_{n\times k}$ defined by
    \begin{align*}
        &C_{S,1:k} = I_{k\times k} - \frac{1}{k} J_{k\times k},\\
        &C_{V\setminus S,1:k} = 0,
    \end{align*}
    the following statement holds:
    \begin{align*}
        \norm{x - \xtil} \leq \eps \cdot 8n^{8} \cdot \kh{\frac{\wmax}{\wmin}}^{2.5},
    \end{align*}
    where %(here subscript $S$ denotes taking coordinates in $S$
    %and $\Vert$ denotes concatenation)
    \begin{align*}
        &x := \kh{q^T W_G^{1/2} B_G L_G^+ C (SC(L_G, S)) C^T L_G^+}^T, \\
        &\xtil := \LaplSolve(L_G, C \xtil_1, \eps), \\
        &\xtil_1 :=
        (L_G)_{S,S} (C^T \xtil_2) - (L_G)_{S,T}\ \LaplSolve((L_G)_{T,T}, (L_G)_{T,S}(C^T \xtil_2), \ep), \\
        &\xtil_2 := \LaplSolve(L_G, B_G^T W_G^{1/2} q, \eps).
    \end{align*}
\end{lemma}

Before proving the above two lemmas, we show how they imply
Lemma~\ref{lem:diff-apx}.
\begin{proof}[Proof of Lemma~\ref{lem:diff-apx}]
    The running time follows directly from the running time of $\LaplSolve$.
    Let $X_{k\times n} := Q W_G^{1/2} B_G L_G^+ C (SC(L_G, S)) C^T L_G^+$.
    The multiplicative approximation follows from
    Johnson-Lindenstrauss Lemma.
    To prove the additive approximation,
    we write the difference between
    $\norm{Xb_e}^2$ and $\norm{Yb_e}^2$
    as
    \begin{align*}
        \abs{ \norm{X b_e}^2 - \norm{Y b_e}^2 } =
        & \abs{ \norm{X b_e} - \norm{Y b_e} } \cdot
        \kh{ \norm{X b_e} + \norm{Y b_e} }.
    \end{align*}
    Let $u,v$ be the endpoints of $e$.
    We upper bound $\abs{ \norm{X b_e} - \norm{Y b_e} }$ by
    \begin{align*}
        \abs{ \norm{X b_e} - \norm{Y b_e} }
        & \leq \norm{(X - Y) b_e} = \norm{(X - Y) (e_u - e_v)}
        \qquad \text{by triangle ineq.} \notag \\
        & \leq \norm{(X - Y) e_u} + \norm{(X - Y) e_v}
        \qquad \text{by triangle ineq.} \notag \\
        & \leq \sqrt{2} \kh{\norm{(X - Y) e_u}^2 + \norm{(X - Y) e_v}^2}^{1/2}
        \qquad \text{by Cauchy-Schwarz}\\
        & \leq \sqrt{2} \norm{X - Y}_F = \sqrt{2} \kh{\sum_{i=1}^{k} \norm{(X - Y)^T e_i}^2}^{1/2} \\
        & \leq \sqrt{2k} \cdot \eps \cdot 8n^{8} \cdot \kh{\frac{\wmax}{\wmin}}^{2.5}
        \qquad \text{by Lemma~\ref{lem:computeemb}} \\
        & \leq \frac{\delta_1}{3\sqrt{2} n^{1/2} \wmin^{-1/2}}
    \end{align*}
    and upper bound $\norm{X b_e} + \norm{Y b_e}$ by
    \begin{align*}
        \norm{X b_e} + \norm{Y b_e} \leq
        & 2\norm{X b_e} + \abs{ \norm{X b_e} - \norm{Y b_e} } \\
        \leq & 2 \kh{ (1 + \delta_0) \kh{b_e^T L_G^+ b_e - b_e^T L_{G/S}^+ b_e} }^{1/2}
        + \abs{ \norm{X b_e} - \norm{Y b_e} }
        \quad \text{by Lemma~\ref{lem:jl}} \\
        \leq & 2 \kh{ (1 + \delta_0) n / \wmin }^{1/2} + \abs{ \norm{X b_e} - \norm{Y b_e} }
        \qquad \text{upper bounding $b_e^T L_G^+ b_e$} \\
        \leq & 3\sqrt{2} n^{1/2} \wmin^{-1/2} \qquad \text{by $\delta_0 < 1$}
    \end{align*}
    Combining these two upper bounds gives
    \begin{align*}
        \abs{ \norm{X b_e}^2 - \norm{Y b_e}^2 } \leq \delta_1,
    \end{align*}
    which proves the additive error.
\end{proof}

\subsection{Analysis of additional errors}
We now prove Lemma~\ref{lem:applyschur} and~\ref{lem:computeemb}.
\begin{proof}[Proof of Lemma~\ref{lem:applyschur}]
    We upper bound $\norm{x - \xtil}$ by
    \begin{align*}
        \norm{x - \xtil}
        = &\norm{(L_G)_{S,T} \kh{(L_G)_{T,T}^{-1} (L_G)_{T,S}b - \xtil_1}} \\
        \leq & n \wmax \norm{(L_G)_{T,T}^{-1} (L_G)_{T,S}b - \xtil_1}
        \qquad \text{by~(\ref{eq:singular3})} \\
        \leq & \eps n^{2.5} \wmax^{1.5} \wmin^{-0.5} \norm{(L_G)_{T,S}b}
        \qquad \text{by Lemma~\ref{lem:mnorm}} \\
        \leq & \eps n^{3.5} \wmax^{2.5} \wmin^{-0.5} \norm{b}
        \qquad \text{by~(\ref{eq:singular3}).}
    %$\sigmamax((L_G)_{S,T}) \leq n^2 \kh{\frac{\wmax}{\wmin}}^{1/2}$}
    \end{align*}
\end{proof}

\begin{proof}[Proof of Lemma~\ref{lem:computeemb}]
    We first bound the norm of vector
    $L_G^+ B_G^T W_G^{1/2} q$ by
    \begin{align}
        \norm{L_G^+ B_G^T W_G^{1/2} q} \leq &\frac{n^2}{\wmin} \norm{q}
        \qquad \text{by $\sigmamax(L_G^+ B_G^T W_G^{1/2}) = \lambdamax(L_G^+)$ and~(\ref{eq:eigen1})} \notag\\
        = &\frac{n^{2.5}}{\wmin}
        \qquad \text{since $q$'s entries are $\pm 1$,} \label{eq:norm1}
    \end{align}
    and upper bound the norm of vector
    $SC(L_G,S) C^T L_G^+ B_G^T W_G^{1/2} q$ by
    \begin{align}
        \norm{SC(L_G,S) C^T L_G^+ B_G^T W_G^{1/2} q} \leq &n \wmax \norm{L_G^+ B_G^T W_G^{1/2} q}
        \qquad \text{by~(\ref{eq:singular2})} \notag \\
        \leq &n^{3.5} \frac{\wmax}{\wmin}. \label{eq:norm2}
    \end{align}
    %\begin{align}
    %    \norm{L_G^+ C SC(L_G,S) C^T L_G^+ B_G^T W_G^{1/2} q} \leq &\frac{n^2}{\wmin}
    %    \norm{SC(L_G,S) C^T L_G^+ B_G^T W_G^{1/2} q}
    %    \qquad \text{by~(\ref{eq:eigen1})} \notag \\
    %    \leq &n^{5.5} \wmax \wmin^{-2} \label{eq:norm3}
    %\end{align}
    The error of $\xtil_2$ follows by
    \begin{align}
        \norm{L_G^+ B_G^T W_G^{1/2} q - \xtil_2}
        \leq &\eps n^{1.5} \kh{\frac{\wmax}{\wmin}}^{1/2} \norm{L_G^+ B_G^T W_G^{1/2} q}
        \qquad \text{by Lemma~\ref{lem:lnorm}} \notag\\
        \leq & \eps n^{4} \wmax^{1/2} \wmin^{-1.5}
        \qquad \text{by~(\ref{eq:norm1}).}
        \label{eq:difx2}
    \end{align}
    The norm of $\xtil_2$ can be upper bounded by
    \begin{align}
        \norm{\xtil_2} \leq &\norm{L_G^+ B_G^T W_G^{1/2} q} + \norm{L_G^+ B_G^T W_G^{1/2} q - \xtil_2}
        \qquad \text{by triangle inequality} \notag\\
        \leq & \frac{2 n^{2.5}}{\wmin}
        \qquad \text{by~(\ref{eq:norm1}) and~(\ref{eq:eps}).} \label{eq:normx2}
    \end{align}
    The error of $\xtil_1$ follows by
    \begin{align}
        & \norm{SC(L_G,S) C^T L_G^+ B_G^T W_G^{1/2} q - \xtil_1} \notag\\
        \leq &\norm{SC(L_G,S) C^T \kh{L_G^+ B_G^T W_G^{1/2} q - \xtil_2}}
        +
        \norm{SC(L_G,S)C^T \xtil_2 - \xtil_1}
        \quad \text{by triangle ineq.} \notag\\
        \leq & \eps n^{5} \wmax^{1.5} \wmin^{-1.5} +
        \norm{SC(L_G,S)C^T \xtil_2 - \xtil_1}
        \qquad \text{by~(\ref{eq:difx2}),~(\ref{eq:singular2}) and $\sigmamax(C) = 1$} \notag\\
        \leq &\eps n^{5} \wmax^{1.5} \wmin^{-1.5} +
        \ep \cdot \norm{C^T \xtil_2} \cdot n^{3.5} \cdot \wmax^{2.5} \wmin^{-0.5}
        \qquad \text{by Lemma~\ref{lem:applyschur}} \notag\\
        \leq &\eps n^{5} \wmax^{1.5} \wmin^{-1.5} +
        \ep \cdot 2n^{6} \cdot \wmax^{2.5} \wmin^{-1.5}
        \qquad \text{by~(\ref{eq:normx2})} \notag\\
        \leq & \eps \cdot 4n^{6} \cdot \wmax^{2.5} \wmin^{-1.5}.
        \label{eq:difx1}
    \end{align}
    The norm of $\xtil_1$ can be upper bounded by
    \begin{align}
        &\norm{\xtil_1} \notag \\ \leq &\norm{SC(L_G,S) C^T L_G^+ B_G^T W_G^{1/2} q} + 
        \norm{SC(L_G,S) C^T L_G^+ B_G^T W_G^{1/2} q - \xtil_1}
        \ \text{by triangle ineq.} \notag\\
        \leq & 2n^{3.5} \frac{\wmax}{\wmin}
        \qquad \text{by~(\ref{eq:norm2}) and~(\ref{eq:eps}).} 
        \label{eq:normx1}
    \end{align}
    Finally, the error of $\xtil$ follows by
    \begin{align}
        &\norm{x - \xtil} \notag\\
        \leq & \norm{L_G^+ C \kh{SC(L_G,S) C^T L_G^+ B_G^T W_G^1/2 q - \xtil_1}}
        + \norm{L_G^+ C \xtil_1 - \xtil}
        \qquad \text{by triangle ineq.} \notag\\
        \leq &
        \eps \cdot 4n^{8} \cdot \wmax^{2.5} \wmin^{-2.5}
        + \norm{L_G^+ C \xtil_1 - \xtil}
        \qquad \text{by~(\ref{eq:difx1}) and~(\ref{eq:eigen1})} \notag \\
        \leq &
        \eps \cdot 4n^{8} \cdot \wmax^{2.5} \wmin^{-2.5}
        + \eps n^{1.5} \wmax^{0.5} \wmin^{-0.5} \norm{C \xtil_1}
        \qquad \text{by Lemma~\ref{lem:lnorm}} \notag \\
        \leq &
        \eps \cdot 4n^{8} \cdot \wmax^{2.5} \wmin^{-2.5}
        + 2\eps n^{5} \wmax^{1.5} \wmin^{-1.5}
        \qquad \text{by~(\ref{eq:normx1})} \notag \\
        \leq &
        \eps \cdot 8n^{8} \cdot \wmax^{2.5} \wmin^{-2.5}
        \label{eq:difx}
    \end{align}
\end{proof}

\section{Better effective resistance approximation}\label{sec:res-apx}

%\todo{expand this section}

%This is the same as in \cite{DKPRS17},
%except that 

In this section, we use divide-and-conquer based on Theorem~\ref{thm:schur-fast}
to $\eps$-approximate effective resistances
for a set of pairs of vertices $P\subseteq V(G)\times V(G)$
in time $O(m^{1 + o(1)} + (\sizeof{P} / \eps^2)\text{polylog}(n))$.
The reduction we use is the same as in~\cite{DKPRS17}.
%except that we perform divide-and-conquer on the queries
%instead of vertices.
We give the algorithm $\ResApx$ as follows:

\begin{algorithm}[H]
    \SetAlgoLined
    \DontPrintSemicolon

    \KwIn{A weighted graph $G$, a set of pairs of vertices $P$, and an $\eps \in (0,1)$}

    \KwOut{Estimates $\left\{\rtil_{u,v}\right\}_{(u,v)\in P}$ to
    effective resistances between vertex pairs in $P$}

    \If{$\sizeof{P} = 1$}{
        Compute the Schur complement $H$ of $G$ onto $P$ with error $\ep$\;
        %using Theorem~\ref{thm:schur-fast}
        %$H \gets \FastSchur(G,\setof{u,v},\ep)$ $\;
        %Compute $L_H^+$ by inverting $L_H$ in $O(1)$ time. \;
        \Return{$\setof{\rtil_{u,v} := b_{u,v}^T L_H^+ b_{u,v}}$} for the only $(u,v) \in P$\; 
    }
    Let $\eps_1 := \frac{1}{2}\cdot \eps \cdot (1 / \log \sizeof{P})$
    and $\eps_2 := \eps \cdot (1 - 1 / \log \sizeof{P})$. \;
    Divide $P$ into subsets $P^{(1)}$ and $P^{(2)}$ with equal sizes. \;
    Let $V^{(1)}$ and $V^{(2)}$ be the respective set of vertices
    in $P^{(1)}$ and $P^{(2)}$. \;
    Compute the Schur complement $H^{(1)}$ of $G$ onto $V^{(1)}$ with error $\ep_1$\;
    Compute the Schur complement $H^{(2)}$ of $G$ onto $V^{(2)}$ with error $\ep_2$\;
    %$H^{(1)} \gets \FastSchur(G, V^{(1)}, \ep_1)$,
    %$H^{(2)} \gets \FastSchur(G, V^{(2)}, \ep_1)$\;
    $\rtil \gets \ResApx(H^{(1)}, P^{(1)}, \ep_2) \cup \ResApx(H^{(2)}, P^{(2)}, \ep_2)$\;
    \Return{$\rtil$}

    \caption{$\ResApx(G,P,\ep)$, never executed}
\end{algorithm}

\begin{proof}[Proof of Corollary~\ref{cor:res-apx}]
    The approximation guarantees follows from
    \begin{align*}
        \rtil_{u,v} \geq
        & \kh{1 - \frac{1}{2}\cdot \eps/\log \sizeof{P}} ^ {\log \sizeof{P} - 1}
        \cdot \kh{b_{u,v}^T L_G^+ b_{u,v}} \\
        \geq & (1 - \eps) b_{u,v}^T L_G^+ b_{u,v}
    \end{align*}
    and
    \begin{align*}
        \rtil_{u,v} \leq
        & \kh{1 + \frac{1}{2}\cdot \eps/\log \sizeof{P}} ^ {\log \sizeof{P} - 1}
        \cdot \kh{b_{u,v}^T L_G^+ b_{u,v}} \\
        \leq & (1 + \eps) b_{u,v}^T L_G^+ b_{u,v}.
    \end{align*}
    We then prove the running time.
    Let $T(p, \eps)$ denote the running time of $\ResApx(G, P, \ep)$
    when $\sizeof{P} = p$ and $\sizeof{E(G)} = O((p/\ep^2)\text{polylog}(n))$.
    Clearly, the total running time of $\ResApx(G,P,\ep)$ for any
    $G$ with $m$ edges is at most
    \begin{align}\label{eq:totrt}
        2\cdot T(\sizeof{P}/2, \eps\cdot (1 - 1 / \log\sizeof{P})) +
        O\kh{m^{1 + o(1)} + (\sizeof{P}/\eps^{2})\text{polylog}(n)},
    \end{align}
    since the first step of $\ResApx$ will divide the graph into
    two Schur complements with $O((\sizeof{P}/\eps^{2})\text{polylog}(n))$ edges each.
    Furthermore,
    we can write $T(p,\eps)$ in a recurrence form as
    \begin{align*}
        T(p,\eps) = 2\cdot T(p/2, \eps \cdot (1 - 1/\log p)) +
        O\kh{p^{1 + o(1)} + (p/\eps^{2})\text{polylog}(n)},
    \end{align*}
    which gives
    \[
        T(p,\eps) = O\kh{p^{1 + o(1)} + (p/\eps^2)\text{polylog}(n)}.
    \]
    Combining this with~(\ref{eq:totrt}) gives the overall running time
    \[
        O\kh{m^{1 + o(1)} + (\sizeof{P} / \eps^2)\text{polylog}(n)}.
    \]
\end{proof}

%\bibliographystyle{alpha}
%\bibliography{sss}

\begin{thebibliography}{CKM{\etalchar{+}}14}

\bibitem[Ach01]{A01}
Dimitris Achlioptas.
\newblock Database-friendly random projections.
\newblock In {\em Proceedings of the 20th {ACM} {SIGACT-SIGMOD-SIGART}
  Symposium on Principles of Database Systems (PODS)}, 2001.

\bibitem[AGK14]{AGK14}
Alexandr Andoni, Anupam Gupta, and Robert Krauthgamer.
\newblock Towards ($1 + \epsilon$)-approximate flow sparsifiers.
\newblock In {\em Proceedings of the Twenty-Fifth Annual {ACM-SIAM} Symposium
  on Discrete Algorithms, {SODA} 2014, Portland, Oregon, USA, January 5-7,
  2014}, pages 279--293, 2014.

\bibitem[CGH16]{CGH16}
Yun~Kuen Cheung, Gramoz Goranci, and Monika Henzinger.
\newblock Graph minors for preserving terminal distances approximately - lower
  and upper bounds.
\newblock In {\em 43rd International Colloquium on Automata, Languages, and
  Programming, {ICALP} 2016, July 11-15, 2016, Rome, Italy}, pages
  131:1--131:14, 2016.

\bibitem[Chu12]{C12}
Julia Chuzhoy.
\newblock On vertex sparsifiers with steiner nodes.
\newblock In {\em Proceedings of the 44th Symposium on Theory of Computing
  Conference, {STOC} 2012, New York, NY, USA, May 19 - 22, 2012}, pages
  673--688, 2012.

\bibitem[CKM{\etalchar{+}}14]{CKMPPRX14}
Michael~B. Cohen, Rasmus Kyng, Gary~L. Miller, Jakub~W. Pachocki, Richard Peng,
  Anup~B. Rao, and Shen~Chen Xu.
\newblock Solving sdd linear systems in nearly mlog1/2n time.
\newblock In {\em Proceedings of the Forty-sixth Annual ACM Symposium on Theory
  of Computing}, STOC '14, pages 343--352, New York, NY, USA, 2014. ACM.

\bibitem[CL06]{CL06}
Fan Chung and Linyuan Lu.
\newblock Concentration inequalities and martingale inequalities: a survey.
\newblock {\em Internet Math.}, 3(1):79--127, 2006.

\bibitem[CLLM10]{CLLM10}
Moses Charikar, Tom Leighton, Shi Li, and Ankur Moitra.
\newblock Vertex sparsifiers and abstract rounding algorithms.
\newblock {\em CoRR}, abs/1006.4536, 2010.

\bibitem[DKP{\etalchar{+}}17]{DKPRS17}
David Durfee, Rasmus Kyng, John Peebles, Anup~B. Rao, and Sushant Sachdeva.
\newblock Sampling random spanning trees faster than matrix multiplication.
\newblock In {\em Proceedings of the 49th Annual ACM SIGACT Symposium on Theory
  of Computing}, STOC 2017, pages 730--742, New York, NY, USA, 2017. ACM.

\bibitem[DPPR17]{DPPR17}
David Durfee, John Peebles, Richard Peng, and Anup~B. Rao.
\newblock Determinant-preserving sparsification of {SDDM} matrices with
  applications to counting and sampling spanning trees.
\newblock {\em CoRR}, abs/1705.00985, 2017.

\bibitem[EGK{\etalchar{+}}14]{EGKRTT14}
Matthias Englert, Anupam Gupta, Robert Krauthgamer, Harald R{\"{a}}cke, Inbal
  Talgam{-}Cohen, and Kunal Talwar.
\newblock Vertex sparsifiers: New results from old techniques.
\newblock {\em {SIAM} J. Comput.}, 43(4):1239--1262, 2014.

\bibitem[GHP17]{GHP17}
Gramoz Goranci, Monika Henzinger, and Pan Peng.
\newblock {Improved Guarantees for Vertex Sparsification in Planar Graphs}.
\newblock In Kirk Pruhs and Christian Sohler, editors, {\em 25th Annual
  European Symposium on Algorithms (ESA 2017)}, volume~87 of {\em Leibniz
  International Proceedings in Informatics (LIPIcs)}, pages 44:1--44:14,
  Dagstuhl, Germany, 2017. Schloss Dagstuhl--Leibniz-Zentrum fuer Informatik.

\bibitem[GHP18]{GHP18}
Gramoz Goranci, Monika Henzinger, and Pan Peng.
\newblock Dynamic effective resistances and approximate schur complement on
  separable graphs.
\newblock {\em CoRR}, abs/1802.09111, 2018.

\bibitem[Ind06]{I06}
Piotr Indyk.
\newblock Stable distributions, pseudorandom generators, embeddings, and data
  stream computation.
\newblock {\em J. {ACM}}, 53(3):307--323, 2006.

\bibitem[JL84]{JL84}
William~B Johnson and Joram Lindenstrauss.
\newblock {Extensions of Lipschitz mappings into a Hilbert space}.
\newblock {\em Contemporary Mathematics}, 26(189-206):1, 1984.

\bibitem[{Kir}47]{K47}
G.~{Kirchhoff}.
\newblock {Ueber die Aufl{\"o}sung der Gleichungen, auf welche man bei der
  Untersuchung der linearen Vertheilung galvanischer Str{\"o}me gef{\"u}hrt
  wird}.
\newblock {\em Annalen der Physik}, 148:497--508, 1847.

\bibitem[KLP{\etalchar{+}}16]{KLPSS16}
Rasmus Kyng, Yin~Tat Lee, Richard Peng, Sushant Sachdeva, and Daniel~A.
  Spielman.
\newblock Sparsified cholesky and multigrid solvers for connection laplacians.
\newblock In {\em Proceedings of the 48th Annual {ACM} {SIGACT} Symposium on
  Theory of Computing, {STOC} 2016, Cambridge, MA, USA, June 18-21, 2016},
  pages 842--850, 2016.

\bibitem[KMST10]{KMST10}
Alexandra Kolla, Yury Makarychev, Amin Saberi, and Shang{-}Hua Teng.
\newblock Subgraph sparsification and nearly optimal ultrasparsifiers.
\newblock In {\em Proceedings of the 42nd {ACM} Symposium on Theory of
  Computing, {STOC} 2010, Cambridge, Massachusetts, USA, 5-8 June 2010}, pages
  57--66, 2010.

\bibitem[KS16]{KS16}
R.~Kyng and S.~Sachdeva.
\newblock Approximate gaussian elimination for laplacians - fast, sparse, and
  simple.
\newblock In {\em 2016 IEEE 57th Annual Symposium on Foundations of Computer
  Science (FOCS)}, pages 573--582, Oct 2016.

\bibitem[LM10]{LM10}
Frank~Thomson Leighton and Ankur Moitra.
\newblock Extensions and limits to vertex sparsification.
\newblock In {\em Proceedings of the 42nd {ACM} Symposium on Theory of
  Computing, {STOC} 2010, Cambridge, Massachusetts, USA, 5-8 June 2010}, pages
  47--56, 2010.

\bibitem[MM10]{MM10}
Konstantin Makarychev and Yury Makarychev.
\newblock Metric extension operators, vertex sparsifiers and lipschitz
  extendability.
\newblock In {\em 51th Annual {IEEE} Symposium on Foundations of Computer
  Science, {FOCS} 2010, October 23-26, 2010, Las Vegas, Nevada, {USA}}, pages
  255--264, 2010.

\bibitem[Moi09]{M09}
Ankur Moitra.
\newblock Approximation algorithms for multicommodity-type problems with
  guarantees independent of the graph size.
\newblock In {\em 50th Annual {IEEE} Symposium on Foundations of Computer
  Science, {FOCS} 2009, October 25-27, 2009, Atlanta, Georgia, {USA}}, pages
  3--12, 2009.

\bibitem[MP13]{MP13}
Gary~L. Miller and Richard Peng.
\newblock Approximate maximum flow on separable undirected graphs.
\newblock In {\em Proceedings of the Twenty-fourth Annual ACM-SIAM Symposium on
  Discrete Algorithms}, SODA '13, pages 1151--1170, Philadelphia, PA, USA,
  2013. Society for Industrial and Applied Mathematics.

\bibitem[RST14]{RST14}
Harald R{\"{a}}cke, Chintan Shah, and Hanjo T{\"{a}}ubig.
\newblock Computing cut-based hierarchical decompositions in almost linear
  time.
\newblock In {\em Proceedings of the Twenty-Fifth Annual {ACM-SIAM} Symposium
  on Discrete Algorithms, {SODA} 2014, Portland, Oregon, USA, January 5-7,
  2014}, pages 227--238, 2014.

\bibitem[Sch17]{Sc17}
Aaron Schild.
\newblock An almost-linear time algorithm for uniform random spanning tree
  generation.
\newblock {\em CoRR}, abs/1711.06455, 2017.

\bibitem[SRS17]{SRS17}
Aaron Schild, Satish Rao, and Nikhil Srivastava.
\newblock Localization of electrical flows.
\newblock {\em CoRR}, abs/1708.01632, 2017.

\bibitem[SS08]{SS08}
Daniel~A. Spielman and Nikhil Srivastava.
\newblock Graph sparsification by effective resistances.
\newblock In {\em Proceedings of the 40th Annual {ACM} Symposium on Theory of
  Computing, Victoria, British Columbia, Canada, May 17-20, 2008}, pages
  563--568, 2008.

\bibitem[ST14]{ST14}
Daniel~A. Spielman and Shang{-}Hua Teng.
\newblock Nearly linear time algorithms for preconditioning and solving
  symmetric, diagonally dominant linear systems.
\newblock {\em {SIAM} J. Matrix Analysis Applications}, 35(3):835--885, 2014.

\bibitem[Tro11]{T11}
Joel Tropp.
\newblock Freedman's inequality for matrix martingales.
\newblock {\em Electron. Commun. Probab.}, 16:262--270, 2011.

\end{thebibliography}

\newcommand{\etalchar}[1]{$^{#1}$}

\begin{appendix}
    \section{Bounds on eigenvalues of Laplacians and SDDM matrices}
\label{sec:eigen}

We first give upper bounds on the traces of the inverses of Laplacians
and their submatrices.

\begin{lemma}\label{lem:traces}
    For any Laplacian $L_G$ and $S\subset V(G)$,
    \begin{align}
        \trace{L_G^+} &\leq n^2/\wmin, \\
        \trace{(L_G)_{S,S}^{-1}} &\leq n^2/\wmin.
    \end{align}
\end{lemma}
\begin{proof}[Proof of Lemma~\ref{lem:traces}]
    Let $T := V(G)\setminus S$.
    The first upper bound follows by
    \begin{align}\label{eq:ubtr}
        \trace{L_G^+} =
        %\frac{1}{n}\trace{L_G^+ (nI - J)}
        %= \frac{1}{n} \trace{L_G^+ \sum_{u,v} b_{u,v} b_{u,v}^T}
        \frac{1}{n} \sum_{u,v\in V} b_{u,v}^T L_G^+ b_{u,v}
        \leq \frac{1}{n} (n ^ 3 \frac{1}{\wmin})
        \leq \frac{n^2}{\wmin}.
    \end{align}
    The second upper bound follows by
    \begin{align}\label{eq:ubtr2}
        \trace{(L_G)_{S,S}^{-1}}
        = \sum\limits_{u\in S} b_{u,T}^T L_{G/T}^+ b_{u,T}
        \leq n\cdot \frac{n}{\wmin} %\sum\limits_{i=1}^{n-1} i\cdot \frac{1}{\wmin}
		\leq \frac{n^2}{\wmin}.
	\end{align}
    The first inequalities of~(\ref{eq:ubtr}) and~(\ref{eq:ubtr2}) both
    follow from the fact that the effective resistance is at most
    the shortest path.
\end{proof}

\lemlambdas*
\begin{proof}[Proof of Lemma~\ref{lem:lambdas}]
    For the upper bounds, we have
	%We upper bound $\lambdamax$ using Cauchy interlacing:
	\begin{align*}
        \lambdamax\kh{(L_G)_{S,S}} \leq
		\lambdamax\kh{L} \leq
		\lambdamax\kh{\wmax L_{K_n}}
		\leq
        n \wmax,
	\end{align*}
    where the first inequality follows from Cauchy interlacing,
    and $K_n$ denotes the complete graph of $n$ vertices.

    For the lower bounds, we have
	%The lower bound of $\lambdamin$ follows by
	\begin{align*}
        \lambda_2(L_G) &\geq 1 / \trace{L_G^+} \geq \wmin / n^2, \\
        \lambdamin\kh{L_{S,S}}
        &\geq 1 / \trace{(L_G)_{S,S}^{-1}} \geq \wmin / n^2.
	\end{align*}
\end{proof}

\section{Bounds on 2-norms of some useful matrices}
\label{sec:2norm}

\lemsigmas*
\begin{proof}[Proof of Lemma~\ref{lem:sigmas}]
    The largest singular value of $W_G^{1/2} B_G$ follows by
    \begin{align*}
        & \sigmamax(W_G^{1/2} B_G) \leq \kh{\lambdamax(L_G)}^{1/2} \leq \kh{n\wmax}^{1/2}
        \qquad \text{by~(\ref{eq:eigen3}).}
    \end{align*}
    The largest eigenvalue of Schur complements follows by
    \begin{align*}
        \lambdamax(SC(L_G,S)) \leq \lambdamax((L_G)_{S,S}) \leq n\wmax
        \qquad \text{by~(\ref{eq:eigen3}).}
    \end{align*}
    The largest singular value of $(L_G)_{S,T}$ follows by
    \begin{align*}
        \sigmamax((L_G)_{S,T}) \leq & \kh{\lambdamax\kh{(L_G)_{S,T}^T (L_G)_{T,S}}}^{1/2} \\
        \leq & \kh{n \wmax \cdot \lambdamax\kh{(L_G)_{S,T}^T (L_G)_{T,T}^{-1} (L_G)_{T,S}}}^{1/2}
        \qquad \text{by~(\ref{eq:eigen3})} \\
        \leq & \kh{n \wmax \cdot \lambdamax\kh{(L_G)_{S,S}}}^{1/2}
        \quad \text{since $SC(L_G,S)$ is positive semi-definite} \\
        \leq & n \wmax \qquad  \text{by~(\ref{eq:eigen3}).}
    \end{align*}
\end{proof}

%\section[Bounding errors of $\LaplSolve$ by $l_2$ norms]{Bounding errors of $\LaplSolve$ by $\ell_2$ norms}
\section{Bounds on errors of $\LaplSolve$ using $\ell_2$ norms}
\label{sec:l2norm}

\lemlnorm*
\begin{proof}[Proof of Lemma~\ref{lem:lnorm}]
    The error follows by
    \begin{align*}
        \norm{x - \xtil} \leq
        & n \wmin^{-1/2} \norm{x - \xtil}_{L_G} \qquad \text{by~(\ref{eq:eigen1})} \\
        \leq & n \wmin^{-1/2} \eps\norm{x}_{L_G} \leq
        n^{1.5} \kh{\frac{\wmax}{\wmin}}^{1/2} \qquad \text{by~(\ref{eq:eigen3})}
    \end{align*}
\end{proof}

\lemmnorm*
\begin{proof}[Proof of Lemma~\ref{lem:mnorm}]
    The error follows by
    \begin{align*}
        \norm{x - \xtil} \leq
        & n \wmin^{-1/2} \norm{x - \xtil}_{M} \qquad \text{by~(\ref{eq:eigen2})} \\
        \leq & n \wmin^{-1/2} \eps\norm{x}_{M} \leq
        n^{1.5} \kh{\frac{\wmax}{\wmin}}^{1/2} \qquad \text{by~(\ref{eq:eigen3})}
    \end{align*}
\end{proof}

    \section{Split subroutines}

    \propsplit*

\begin{algorithm}[H]
\SetAlgoLined
\DontPrintSemicolon

    \KwIn{a graph $H$}

    \KwOut{a graph $I$ with a pair of edges for each edge in $H$ and a set of paired edges in $\mc P$}

    $I\gets H$\;

    $\mc P\gets \emptyset$\;

    \ForEach{edge $e\in E(H)$}{

        \uIf{1/16-JL-approximation to $\texttt{lev}_H(e)\ge 1/2$}{

            Replace $e = \{u,v\}\in E(I)$ with two edges $e_0 = \{u,v\}$ and $e_1 = \{u,v\}$ with $r_{e_0} = r_{e_1} = 2r_e$\;

            Add the pair $(e_0,e_1)$ to $\mc P$\;

        }\Else{

            Add a vertex $w$ to $V(I)$\;

            Replace $e = \{u,v\}\in E(I)$ with two edges $e_0 = \{u,w\}$ and $e_1 = \{w,v\}$ with $r_{e_0} = r_{e_1} = r_e/2$\;

            Add the pair $(e_0,e_1)$ to $\mc P$\;

        }

    }

    \Return{$(I,\mc P)$}

\caption{$\Split(H)$}
\end{algorithm}

\begin{proof}

\textbf{Electrical equivalence.} Two parallel edges with resistance $2r_e$ are electrically equivalent to one edge with resistance $r_e$. Two edges with resistance $r_e/2$ in series are equivalent to one edge with resistance $r_e$. Therefore, both ways of replacing edges in $H$ with pairs of edges in $I$ result in an electrically equivalent graph.

\textbf{Bounded leverage scores.} For an edge $e$ that is replaced with two series edges $e_0$ and $e_1$, $$\texttt{lev}_I(e_0) = \texttt{lev}_I(e_1) = \frac{1}{2} + \frac{\texttt{lev}_H(e)}{2} \in [1/2,3/4]$$ since $\texttt{lev}_H(e)\in [0,1/2(1 + 1/16)]$. For an edge $e$ that is replaced with two parallel edges $e_0$ and $e_1$, $$\texttt{lev}_I(e_0) = \texttt{lev}_I(e_1) = \texttt{lev}_H(e)/2 \in [1/4,1/2]$$ since $\texttt{lev}_H(e)\in [1/2(1 - 1/16),1]$. Since all edges in $I$ result from one of these operations, they all have leverage score in $[3/16,13/16]$, as desired.

\textbf{$\mc P$ description.} (a) describes edges resulting from parallel replacements, while (b) describes edges reesulting from series replacements.

\textbf{Runtime.} Estimating the leverage scores takes near-linear time \cite{SS08}. Besides this, the algorithm just does linear scans of the graph. Therefore, it takes near-linear time.
\end{proof}

\propunsplit*

\begin{algorithm}[H]
\SetAlgoLined
\DontPrintSemicolon

    \KwIn{a graph $I$ and a set of nonintersecting pairs of edges $\mc P$}

    \KwOut{a graph $H$ with each pair unsplit to a single edge}

    $H\gets I$

    \ForEach{pair $(e_0,e_1)\in \mc P$}{

        \uIf{$e_0$ and $e_1$ have the same endpoints $\{u,v\}$ and $e_0,e_1\in E(I)$}{

            Replace $e_0$ and $e_1$ in $H$ with one edge $e = \{u,v\}$ with $r_e = 1/(1/r_{e_0} + 1/r_{e_1})$\;

        }\ElseIf{$e_0 = \{u,w\}$, $e_1 = \{w,v\}$, $w$ has degree 2, and $e_0,e_1\in E(I)$}{

            Replace $e_0$ and $e_1$ in $H$ with one edge $e = \{u,v\}$ with $r_e = r_{e_0} + r_{e_1}$\;

        }

    }
    
\caption{$\Unsplit(I,\mc P)$}
\end{algorithm}

\begin{proof}

\textbf{Electrical equivalence.} Two parallel edges with resistance $r_{e_0}$ and $r_{e_1}$ are electrically equivalent to one edge with resistance $1/(1/r_{e_0} + 1/r_{e_1})$. Two edges with resistance $r_{e_0}$ and $r_{e_1}$ in series are equivalent to one edge with resistance $r_{e_0} + r_{e_1}$. Therefore, both ways of replacing pairs of edges in $I$ with single edges in $H$ result in an electrically equivalent graph.

\textbf{Edges of $H$.} Since the pairs in $\mc P$ do not intersect, the map $\phi(e_i) = e$ that maps an edge $e_i$, $i\in \{0,1\}$ to the $e$ as described in the foreach loop is well-defined. Since each $(e_0,e_1)\in \mc P$ pair is assigned to the same edge $e$, $\phi(e_0) = \phi(e_1) = e$. Each edge in the output graph $H$ originates from the initialization of $H$ to $I$, the if statement, or the else statement. These are type (a),(b), and (c) edges respectively. Therefore, $\phi$ satisfies the required conditions. 

\textbf{Runtime.} The algorithm just requires a constant number of linear scans over the graph.

\end{proof}

\end{appendix}

\end{document}